\documentclass[runningheads]{llncs}
\usepackage{graphicx}
\usepackage{todonotes}
\usepackage{latexsym,amssymb,amsmath,amsfonts}

  { }

\newcommand{\diff}{\ensuremath{\mathtt{diff}}}
\newcommand{\el}{\ensuremath{el}}

\newcommand{\len}[1]{\ensuremath{\vert {#1} \vert}\xspace}
\newcommand{\ARRAY}{\ensuremath{\mathtt{ARRAY}}\xspace}
\newcommand{\INDEX}{\ensuremath{\mathtt{INDEX}}\xspace}
\newcommand{\ELEM}{\ensuremath{\mathtt{ELEM}}\xspace}
\newcommand{\formulae}{formul\ae\xspace}

\newcommand{\LRA}{\ensuremath{\mathcal{LRA}}\xspace}
\newcommand{\EUF}{\ensuremath{\mathcal{EUF}}\xspace}
\newcommand{\LIA}{\ensuremath{\mathcal{LIA}}\xspace}
\newcommand{\IDL}{\ensuremath{\mathcal{IDL}}\xspace}

\newcommand{\AXEXT}{\ensuremath{\mathcal{AR}_{{\rm ext}}}\xspace}

\newcommand{\imp}{\rightarrow}

\newcommand{\dpll}[1]{{\sc DPLL}\xspace}

\newcommand{\COMMENT}[1]{}

\newcommand{\alocal}{\ensuremath{A}-local\xspace}
\newcommand{\blocal}{\ensuremath{B}-local\xspace}
\newcommand{\astrict}{\ensuremath{A}-strict\xspace}
\newcommand{\bstrict}{\ensuremath{B}-strict\xspace}
\newcommand{\abcommon}{\ensuremath{AB}-common\xspace}

\newcommand{\ut}{\ensuremath{\underline t}}
\newcommand{\ux}{\ensuremath{\underline x}}

\newcommand{\uy}{\ensuremath{\underline y}}
\newcommand{\uz}{\ensuremath{\underline z}}

\newcommand{\cA}{\ensuremath \mathcal A}
\newcommand{\cB}{\ensuremath \mathcal B}
\newcommand{\cM}{\ensuremath \mathcal M}
\newcommand{\cN}{\ensuremath \mathcal N}

\newcommand{\cI}{\ensuremath \mathcal I}

\renewcommand{\int}{\ensuremath {\mathcal I}}

\newcommand{\Const}{\ensuremath{\mathtt{Const}}\xspace}

\newcommand{\AXD}{\ensuremath{\mathcal{CARD}}\xspace}
\newcommand{\AXDTI}{\ensuremath{\mathcal{\AXD}(T_I)}\xspace}
\newcommand{\estesa}{\ensuremath{\mathcal{CARDC}(T_I)}\xspace}
\newcommand{\AXEXTTI}{\ensuremath{\mathcal{\AXEXT}(T_I)}\xspace}
\newcommand{\extestesa}{\ensuremath{\mathcal{CARC}_{{\rm ext}}(T_I)}\xspace}

\newcommand{\uguale}{\ensuremath{=}\xspace}
\newcommand{\coincide}{\ensuremath{\equiv}\xspace}

\newcommand{\eop}{$\hfill \dashv$}

\usepackage[]{algorithm2e}
\usepackage{framed}
\usepackage{tabularx}
\usepackage{tikz-cd}
\usepackage{wrapfig}
\usepackage{comment}
\usepackage{xspace}
\usepackage{textgreek}
\usepackage{paralist}
\usepackage{enumerate}
\usepackage{amsmath}

\begin{document}
\title{General Interpolation and Strong Amalgamation for Contiguous Arrays}
%
%
\author{Silvio Ghilardi\inst{1}
\and
Alessandro Gianola\inst{2}
\and
Deepak Kapur\inst{3}
\and 
Chiara Naso\inst{1}
}
\authorrunning{S. Ghilardi et al.}
%
\institute{%
 Dipartimento di Matematica, Universit\`a degli Studi di Milano (Italy)\\
\and 
Faculty of Computer Science, Free University of Bozen-Bolzano (Italy)\\
\email{gianola@inf.unibz.it}
\and
Department of Computer Science, University of New Mexico (USA)\\
}
\maketitle              
\begin{abstract}
 Interpolation is an essential tool in software verification, where first-order theories are used to constrain datatypes 
manipulated by programs. 
In this paper, we introduce the datatype theory of \emph{contiguous arrays with maxdiff}, where arrays are completely defined in their allocation memory and for which maxdiff returns the max index where they differ. This theory is strictly more expressive than the array theories previously studied. By showing via an algebraic analysis that its models \emph{strongly amalgamate}, we prove that this theory admits quantifier-free interpolants and, notably, that interpolation transfers to theory combinations. Finally, we provide an algorithm that significantly improves the ones 
for related array theories: 
it  relies on  a polysize  reduction to  general interpolation in linear arithmetics, thus avoiding  
 impractical full terms instantiations and unbounded loops.

 \end{abstract}

\section{Introduction}\label{sec:intro}


Craig Interpolation Theorem \cite{Craig} is a well-known result in first logic that,  given an entailment between two logical formulae $\alpha$ and $\beta$, establishes the existence of a third formula $\gamma$ that shares its non-logical symbols with both $\alpha$ and $\beta$ and such that it is entailed by $\alpha$ and entails $\beta$. Studying interpolation has a long-standing tradition in non-classical logics and in algebraic logic.
Nevertheless, interpolation has been obtaining an increasing attention in automated reasoning and formal verification since the seminal works by McMillan \cite{McM03,McM06}. Indeed, specifically in infinite-state model checking, where an exhaustive, explicit exploration of the state space is not possible, computing interpolants has been proven to be a useful method for practical and efficient approximations of preimage computation.   In the context of software verification, the initial configurations and the transitions relation are usually represented symbolically by means of logical formulae, which gives the possibility of implicitly encoding the execution traces of the system.  More precisely, if $T$ is the first order theory that constraints the state space, one can symbolically express via a suitable $T$-inconsistent formula the fact that the system, starting from its initial configuration, cannot reach  in $n$-steps an error configuration.  Through this inconsistency, interpolants are then extracted from the symbolic representations of these `error' traces with the goal of helping the search of (safety) invariants of the modeled system: interpolants can be successfully used to refine and improve the construction of the candidate safety invariants.

 Model-checking applications usually require that such computed interpolants are not arbitrary but present specific shapes so as to guarantee their concrete usability. Since in many cases studied in software verification the underlying theories have a decidable quantifier-free fragment (but are undecidable or have prohibitive complexity outside),  the most naturale choice is to consider \emph{quantifier-free} interpolants. However, even in case $\alpha$ and $\beta$ are quantifier-free, Craig's Theorem does not guarantee that an interpolant $\gamma$ is quantifier-free too. Indeed, this property, called 'quantifier-free interpolation', does not hold in general for arbitrary first order theories. It is then a non-trivial (and, very often, challenging) problem to prove that useful theories admit 
 quantifier-free interpolation.
 
 In this paper, we are interested in studying the problem of quantifier-free interpolation for an expressive datatype theory that strictly extends the well-studied McCarthy's theory of arrays with extensionality. The original theory was introduced by McCarthy in~\cite{mccarthy}: however, in ~\cite{KMZ06} it is shown that quantifier-free interpolation fails for this theory.  Moreover, although its quantifier-free fragment is decidable, it is well-known that this theory in its full generality is undecidable~\cite{BMS06}: nonetheless, in the same paper, the authors studied a significant decidable fragment, the so-called `array property fragment', which strictly extends the quantifier-free one. The array property fragment 
 is expressive enough to formalize several benchmarks; 
 however, as proved in~\cite{HoeSchi19}, it is not closed under interpolation. Thus, 
 a particularly challenging but interesting problem is 
 that of identifying expressive extensions of the quantifier-free fragment of arrays that are still decidable but also enjoy interpolation: this is what we attack in this contribution.

A first attempt in this direction is in~\cite{lmcs_interp}, where  a
variant of McCarthy's theory was introduced 
 by Skolemizing the axioms of extensionality. This variant turned out to enjoy quantifier-free
interpolation~\cite{lmcs_interp},\cite{TW16}.
However, this Skolem function $\diff$ 
is generic because its semantic interpretation is 
undetermined. Moreover, all the array theories mentioned so far 
allow unlimited out-of-bound write operations and so
cannot express the notion of array \emph{length}, which is fundamental when  formalizing the real behavior of programs. In this respect, there are two possible variants that can be considered: 
\begin{inparaenum}[(i)]
\item the \emph{weak length} $\tilde{\len{-}}$ formalizes the \emph{minimal} interval $[0,\tilde{\len{a}}]$ of indexes  outside which the array $a$ is undefined ($a$ can be undefined also in some location inside $[0,\tilde{\len{a}}]$);
\item the \emph{(strong) length} $\len{-}$ represents the \emph{exact} interval $[0,\len{a}]$ where the array $a$ is fully defined. 
\end{inparaenum}
Strong length is essential for the \emph{faithful} logical formalization of benchmarks coming from software verification, such as C programs included in the SV-COMP competition~\cite{svcomp}. 

These are the main reasons why in \cite{ourfossacs}  
the
theory has been further enriched. There, the semantics of $\diff$, called maxdiff,  
is uniquely determined  in the models of the theory and is
more informative: it  returns \emph{the biggest index} where two different arrays differ. In this theory, weak length can be defined: this is notable, since it represents a first step toward capturing real program arrays.
The main contribution of  \cite{ourfossacs} is to show that this enriched theory has quantifier-free interpolants and its quantifier-free fragment is decidable. Still,  some expressive limitations 
(shared with  the previous literature) persist: arrays are not forced to be completely defined inside their allocation interval (when an array satisfies this property, we call it  `contiguous'), 
because they might contain undefined values in some location. Hence, strong length cannot be defined. Moreover, although in~\cite{ourfossacs} a complete terminating procedure for computing interpolants is provided, 
a complexity upper bound is given only in the simple 
basic case where indexes are just linear orders: 
for more complex 
arithmetical theories of indexes, no complexity analysis  is carried out and the algorithm becomes quite impractical, since it requires to fully instantiate universal quantifiers coming from the theory axioms with index terms of arbitrarily large size.

In this paper, we overcome all those limitations. For that purpose, we introduce the very expressive theory of \emph{contiguous arrays with maxdiff} $\AXDTI$ (parameterized over an index theory $T_I$), which improves and strictly extends the theory presented in~\cite{ourfossacs} by requiring arrays to be all contiguous. This makes the theory more adequate to represent arrays  used in common programming languages: for instance, strong length is now definable. Moreover, in contrast to ~\cite{ourfossacs} where only amalgamation is shown, we prove here a \emph{strong amalgamation} result, when  $\AXDTI$ is enriched with `constant arrays' of a fixed length with a default value in all their locations. Notably, this not only yields that quantifier-free interpolants exist, 
 but also that interpolation is preserved under disjoint signatures combinations and holds in presence of free function symbols (see the definition of `general interpolation' below). This result is completely novel and particularly challenging to be proven, since it requires a sophisticated model-theoretic machinery and a careful algebraic analysis of the class of all models.
 %
 %
We also radically re-design the interpolation algorithm, avoiding the use of unbounded loops and of impractical full instantiation routines. Our new algorithm  reduces the computation of interpolants of a jointly unsatisfiable pair of constraints  to a polynomial size  instance of the same problem in the underlying index theory enriched with unary function symbols. As such, the new algorithm becomes part of the hierarchical interpolation algorithms family~\cite{SS08} and in particular formally resembles the algorithm presented in~\cite{TW16} for array theory enriched with the basic \diff\ symbol. We underline that one aspect  making our problems technically  more challenging than similar problems investigated in the literature is the fact that we handle a combination  with 
\emph{very expressive}
index theories: such a combination is \emph{non-disjoint} because the  total orderings on indexes enter into the specification of the maxdiff and length axioms for arrays.

\subsection{Plan of the paper}\label{subsec:road}
In the following, we call $\EUF$ the  theory of equality and uninterpreted symbols. 
We introduce two novel theories in Section~\ref{sec:tharr}: \AXDTI, i.e., the theory of contiguous arrays with maxdiff, and \estesa, which is an extension of \AXDTI also containing `constant arrays' of a fixed length with a default value (called `$el$') in all locations.
The main technical results 
are that, for every index theory $T_I$:
\begin{compactenum}
 \item[{\rm (i)}] \estesa has general quantifier-free interpolation;
 \item[{\rm (ii)}] \AXDTI enjoys quantifier-free interpolation and 
 such interpolants 
 can be computed hierarchically by relying on a black-box 
  interpolation algorithm for the weaker theory $T_I\cup \EUF$ (which has quantifier free interpolation because $T_I$ is strongly amalgamable, see Theorem~\ref{thm:interpolation-amalgamation}).
\end{compactenum}

Result (i) 
is proved semantically, i.e., we show the equivalent 
strong amalgamation property (see Section~\ref{sec:background} for the definitions). 
The semantic proof 
 requires dedicated constructions (Section~\ref{sec:strongamalg}), 
 relying on some 
 important facts about models 
 and their embeddings (Section~\ref{sec:embeddings}).

The fact that \AXDTI has interpolants follows from the results in 
Section~\ref{sec:strongamalg} 
(where we prove that this theory is amalgamable). 
Result (ii) is proved last (Section~\ref{sec:algo});
we first need an
investigation on the solvability of the $SMT(\AXDTI)$ problem (Section~\ref{sec:sat}). 

We supply here some intuitions about our interpolation algorithm from Section~\ref{sec:algo}.
The algorithm computes an interpolant out of a pair of 
(suitably preprocessed)
mutually unsatisfiable quantifier-free \formulae\ $A^0, B^0$.
We call \emph{common} the variables occurring in both $A^0$ and $B^0$.
The 
existence of quantifier-free interpolants intuitively means that there are two reasoners, one for $A^0$ and one for $B^0$, the first (the second, resp.) of which operates on formulae involving only variables from $A_0$ ($B_0$, resp.). 
The reasoners discover the inconsistency of $A^0\land B^0$ by exchanging information on the common language, i.e., by communicating each other only the entailed quantifier-free formulae 
over the common variables.
The information exchange is hierarchical, i.e., it is limited to $T_I\cup \EUF$-\formulae: literals from the richer language of $\AXDTI$ and outside the language of $T_I\cup \EUF$ can contribute to the information exchange only via  \emph{instantiation} of the universal quantifiers in suitable 
$T_I\cup \EUF$-\formulae given in Section~\ref{sec:tharr}: these formulae, as proved in Lemmas~\ref{lem:univinst} and~\ref{lem:elim}, 
supply equivalent definitions of such literals.
 In contrast to~\cite{ourfossacs}, instantiations of universal quantifiers is 
 limited to variables  and constants for efficiency.

The main problem is to show that the above limited information exchange is sufficient.
This is the case thanks to the fact that the the algorithm manipulates 
\emph{iterated diff operators}~\cite{TW16},\cite{ourfossacs}
(formally defined in Section~\ref{sec:tharr}) and it 
 \emph{gives names} 
to all
such operators when applied to common array variables.
%
Both the production of names for iterated diff-terms and the variable instantiations of the universal quantifiers in the equivalent universal $T_I\cup \EUF$-\formulae
need in principle 
to be repeated infinitely many times;
what we prove  (this is the content of our main Theorem~\ref{thm:al} below) 
is that \emph{a pre-determined polynomial size subset of such manipulations is sufficient}
for the $T_I\cup \EUF$-interpolation module to produce the interpolant we are looking for.\footnote{One could reformulate this fact using the $W$-separability framework from~\cite{TW16}; however, using this framework would not sensibly modify the proof of Theorem~\ref{thm:al}, so we preferred for space reasons and for simplicity to supply 
proofs within standard direct terminology.
}

\subsubsection*{Related work.} We already mentioned the related work on first-order  theories axiomatizing arrays~\cite{mccarthy,KMZ06,lmcs_interp,ourfossacs}, which  our theories of contiguous arrays strictly extend.  Since we adopt a hierarchical approach, our method is closely related to hierarchical interpolation, where interpolants are computed by reduction to a base theory treated as black-box. 
A non-exhaustive summary of this literature is given by the approach in~\cite{RSS07,RS10}, where in the context of linear arithmetic general interpolation is reduced to constraint solving, the one based on \emph{local extensions} in~\cite{SS06,SS08,SS16,SS08} and the one based on $W$-compatibility and finite instantiations of~\cite{TotlaW13,TW16}.


\section{Formal Preliminaries}
\label{sec:background}

We assume the usual syntactic (e.g., signature, variable, term, atom,
literal, formula, and sentence) and semantic (e.g., structure,
sub-structure, truth) 
notions of
first-order logic. 
The equality symbol
``\uguale'' is 
in all signatures. 
Notations like $E(\ux)$ mean that the expression (term, literal,
formula, etc.) $E$ contains free variables only from the tuple $\ux$.
A `tuple of variables' is a list of variables without repetitions and
a `tuple of terms' is a list of terms (possibly with repetitions).
These conventions are useful for substitutions: we use them when denoting with $\phi(\ut/\ux)$ 
(or simply with $\phi(\ut)$) the formula obtained from $\phi(\ux)$ by simultaneous replacement of the `tuple of variables' $\ux$ with the `tuple of terms' $\ut$.
A \emph{constraint} is a conjunction of literals.
A formula is
\emph{universal} (\emph{existential}) iff it is obtained from a
quantifier-free formula by prefixing it with a string of universal
(existential, resp.)  quantifiers.   

\paragraph{Theories and satisfiability modulo theory.}
 A \emph{theory} $T$ is a pair $({\Sigma},
Ax_T)$, where $\Sigma$ is a signature and $Ax_T$ is a set of
$\Sigma$-sentences, called the \emph{axioms} of $T$ (we shall
sometimes write directly $T$ for $Ax_T$).  The \emph{models} of $T$
are those $\Sigma$-structures in which all the sentences in $Ax_T$ are
true.  
A $\Sigma$-formula $\phi$ is \emph{$T$-satisfiable} 
(or $T$-consistent)
if there exists a
model $\cM$ of $T$ such that $\phi$ is true in $\cM$ under a suitable
assignment $\mathtt a$ to the free variables of $\phi$ (in symbols,
$(\cM, \mathtt a) \models \phi$); it is \emph{$T$-valid} (in symbols,
$T\vdash \varphi$) if its negation is $T$-unsatisfiable or,
equivalently, $\varphi$ is provable from the axioms of $T$ in a
complete calculus for first-order logic.  
A theory $T=(\Sigma, Ax_T)$ is \emph{universal} iff 
all sentences in $Ax_{T}$ are
universal.
A formula $\varphi_1$ \emph{$T$-entails} a formula $\varphi_2$ if
$\varphi_1 \to \varphi_2$ is \emph{$T$-valid} (in symbols,
$\varphi_1\vdash_T \varphi_2$ or simply $\varphi_1 \vdash \varphi_2$
when $T$ is clear from the context). 
If $\Gamma$ is a set of \formulae and $\phi$ a formula, $\Gamma \vdash_T \phi$ means that there are 
$\gamma_1, \dots, \gamma_n\in \Gamma$ such that $\gamma_1\wedge \cdots\wedge \gamma_n \vdash_T \phi$.  
 The \emph{satisfiability modulo
  the theory $T$} (SMT$(T)$) \emph{problem} amounts to establishing
the $T$-satisfiability of quantifier-free $\Sigma$-\formulae (equivalently, 
the $T$-satisfiability of  $\Sigma$-constraints).
Some theories have special names, 
which are becoming standard in SMT-literature,
we shall recall some of them during the paper. 
As already mentioned,
we shall call $\EUF(\Sigma)$ (or just \EUF)
the pure equality theory in the signature $\Sigma$.
A theory $T$ admits \emph{quantifier-elimination} iff for every
formula $\phi(\ux)$ there is a quantifier-free formula
$\phi'(\ux)$ such that $T\vdash \phi \leftrightarrow \phi'$. 

\paragraph{Embeddings and sub-structures} 

 The support 
 of a structure $\mathcal{M}$ is denoted
with $|\mathcal{M}|$. For a (sort, constant, function, relation) symbol $\sigma$, we denote as 
$\sigma^\cM$ the interpretation of $\sigma$ in $\cM$.
Let $\cM$ and $\cN$ be two $\Sigma$-structures; 
a $\Sigma$-embedding (or, simply, an embedding) $\mu: \cM \longrightarrow \cN$ is an injective 
function from $|\mathcal{M}|$ into $|\mathcal{N}|$ that preserves and reflects the interpretation of functions 
and relation symbols (see, e.g.,~\cite{CK} for the formal definition). If such an embedding is a set-theoretical 
inclusion, 
we say that $\cM$ is
a {\it substructure} of $\cN$ or that $\cN$ is a {\it superstructure}
of $\cM$.
As it is known,
 the truth of a universal (resp. existential)
sentence is preserved through substructures (resp. superstructures).

Given a signature $\Sigma$ and a $\Sigma$-structure $\cM$, we  
indicate with
$\Delta_{\Sigma}(\cM)$  
the \emph{diagram} of $\cM$:
this is the set of sentences obtained by first expanding $\Sigma$ with a fresh constant $\bar a$ for every
element $a$ from $\vert \cM\vert$ and then taking the set of ground $\Sigma\cup \vert \cM\vert$-literals which are true in $\cM$
(under the  natural expanded interpretation mapping $\bar a$ to $a$).
An easy but nevertheless important basic result (to be frequently used in our proofs), called 
\emph{Robinson Diagram Lemma}~\cite{CK},
says that, given any $\Sigma$-structure $\cN$, there is an embedding $\mu: \cM \longrightarrow \cN$
iff $\cN$ can be expanded to a  $\Sigma\cup \vert \cM\vert$-structure in such a way that it becomes a model of 
$\Delta_{\Sigma}(\cM)$.

\paragraph{Combinations of
    theories.}
A theory $T$ is \emph{stably infinite} iff every $T$-satisfiable
quantifier-free formula (from the signature of $T$) is satisfiable in
an infinite model of $T$.  By compactness, it is possible to show that
$T$ is {stably infinite} iff every model of $T$ embeds into an
infinite one (see, e.g., \cite{Ghil05}).
Let $T_i$ be a stably-infinite theory over the signature $\Sigma_i$
such that the $SMT(T_i)$ problem is decidable for $i=1,2$ and
$\Sigma_1$ and $\Sigma_2$ are disjoint (i.e., the only shared symbol
is equality).  Under these assumptions, the \emph{Nelson-Oppen combination
result}~\cite{NO79} 
says 
that the SMT problem for the combination
$T_1\cup T_2$ of the theories $T_1$ and $T_2$ 
 is
decidable.   Nelson-Oppen result trivially extends to many-sorted languages.

\paragraph{Interpolation properties.}
%
In the introduction, we roughly stated Craig's interpolation theorem~\cite{CK}. 
In this paper, we are interested to
specialize this result to the computation of quantifier-free
interpolants modulo (combinations of) theories.
\begin{definition} \label{def:restricted}[Plain quantifier-free interpolation]
  A theory $T$ \emph{admits (plain) quantifier-free interpolation}
  iff for
  every pair of quantifier-free formulae $\phi, \psi$ such that
  $\psi\wedge \phi$ is $T$-unsatisfiable, there exists a
  quantifier-free formula $\theta$, called an \emph{interpolant}, such
  that: (i) $\psi$ $T$-entails $\theta$, (ii) $\theta \wedge \phi$ is
  $T$-unsatisfiable, and (iii) only the variables occurring in both
  $\psi$ and $\phi$ occur in $\theta$.
\end{definition}

In verification, the following extension of the above definition is
considered more useful.
\begin{definition}
  \label{def:general} [General quantifier-free interpolation]
  Let $T$ be a theory in a signature $\Sigma$; we say that $T$ has the
  \emph{general quantifier-free interpolation property} iff for every
  signature $\Sigma'$ (disjoint from $\Sigma$) and for every pair of
  ground $\Sigma\cup\Sigma'$-\formulae $\phi, \psi$ such that
  $\phi\wedge \psi$ is $T$-unsatisfiable,\footnote{By this (and
    similar notions) we mean that $\phi\wedge \psi$ is unsatisfiable
    in all $\Sigma'$-structures whose $\Sigma$-reduct is a model of
    $T$.  }  there is a ground formula $\theta$ such that: (i) $\phi$
  $T$-entails $\theta$; (ii) $\theta\wedge \psi$ is $T$-unsatisfiable;
  (iv) all 
  relations, constants and function symbols from $\Sigma'$
  occurring in $\theta$ also occur in $\phi$ and $\psi$.
\end{definition}
By replacing free variables with free constants, 
it is easily seen
that the general quantifier-free interpolation property
(Definition~\ref{def:general}) implies the plain quantifier-free
interpolation property (Definition~\ref{def:restricted}); the converse implication does not hold, however
(a counterexample can be found  in this paper too, see Example~\ref{example1} below).

\paragraph{Amalgamation and strong amalgamation.}

Interpolation can be characterized semantically via amalgamation.
\begin{definition}
  \label{def:sub-amalgamation-and-strong-sub-amalgamation}
  A universal theory $T$ has the \emph{amalgamation property} iff, 
  given models $\cM_1$ and $\cM_2$ of $T$ and a common
  submodel $\cA$ of them, there exists a further model $\cM$ of
  $T$ (called $T$-amalgam) 
  endowed with embeddings $\mu_1:\cM_1 \longrightarrow \cM$ and
  $\mu_2:\cM_2 \longrightarrow \cM$ whose restrictions to $|\cA|$
  coincide.

  A universal theory $T$ has the \emph{strong amalgamation property} \cite{kiss} 
   if the
  above embeddings $\mu_1, \mu_2$ and the above model $\cM$ can be chosen so to satisfy the following additional
  condition: if for some $m_1\in \vert \cM_1\vert , m_2\in \vert \cM_2\vert$ we have $\mu_1(m_1)=\mu_2(m_2)$,
  then there exists an element $a$ in $|\cA|$ such that $m_1=a=m_2$.
\end{definition}

The first point 
 of the following theorem is an old result due to~\cite{amalgam}; the second point 
 is proved in~\cite{BGR14} (where it is also suitably reformulated for theories which are not universal):

\begin{theorem}
  \label{thm:interpolation-amalgamation} Let $T$ be a
universal theory. Then 
\begin{description}
 \item[{\rm (i)}] $T$ has the amalgamation property iff it admits
  quan\-ti\-fier-free interpolants;
  \item[{\rm (ii)}] $T$ has the strong amalgamation property iff it has the general quantifier-free interpolation property.
\end{description} 
\end{theorem}

We underline that, in presence of stable infiniteness, strong amalgamation is a modular property (in the sense that it transfers to signature-disjoint unions of theories), whereas amalgamation is not (see again~\cite{BGR14} for details). 
As a special case,
since \EUF has strong amalgamation and is stably infinite, 
the following result follows:

\begin{theorem}\label{thm:ti+euf}
 If $T$ is stably infinite and has strong amalgamation, so does $T\cup \EUF$.
\end{theorem}


%
\section{Arrays with MaxDiff}
\label{sec:tharr}

The  \emph{McCarthy theory of arrays}~\cite{mccarthy} has three
sorts $\ARRAY, \ELEM, \INDEX$ (called ``array'', ``element'', and
``index'' sort, respectively) and two function symbols $rd$ (``read'') and $wr$ (``write'')  
of appropriate arities; its axioms are:
\begin{eqnarray*}
  \forall y, i, e. & & rd(wr(y,i,e),i) \uguale e \\
  \forall y, i, j, e. & & i \not\uguale j \imp rd(wr(y,i,e),j)\uguale rd(y,j) .
\end{eqnarray*}
Arrays \emph{with extensionality} 
have the further
axiom
\begin{eqnarray}\label{ext}
 \forall x, y. 
 x \not\uguale y \imp (\exists i.\ rd(x,i)\not\uguale rd(y,i)),
\end{eqnarray}
called the `extensionality' axiom. This theory is not universal and does not have quantifier-free interpolants. 
Here, we want to  introduce a variant of this theory where Axiom~\eqref{ext} is skolemized via a function $\diff$ with a precise semantic interpretation: it returns \emph{the biggest index} where two different arrays differ. 
We first need the notion of index theory.

\begin{definition}\cite{ourfossacs}\label{def:index}
 An \emph{index theory} $T_I$ is a mono-sorted theory (\INDEX is its sort) satisfying the following conditions:
 \begin{compactenum}
  \item[-] $T_I$ is universal, stably infinite and 
  has the general quantifier-free interpolation property (i.e., it is strongly amalgamable, see Theorem~\ref{thm:interpolation-amalgamation});
  \item[-] $SMT(T_I)$ is decidable;
  \item[-] $T_I$ extends the theory $TO$ of linear orderings with a distinguished  element $0$.
 \end{compactenum}
\end{definition}
We recall that $TO$ is the theory whose only proper symbols (beside equality) are a
 binary predicate $\leq$ and a constant $0$ subject to the  axioms saying that $\leq$ is reflexive, transitive, antisymmetric and total. 
 Thus, the signature of 
 $T_I$ contains at least the binary relation symbol $\leq$
  and the constant $0$. In the paper, when we speak of a $T_I$-term, $T_I$-atom, $T_I$-formula, etc. we mean a term, atom, formula in the signature of $T_I$.
  %
 Below, we use the abbreviation $i<j$ for $i\leq j \land i\not \uguale j$. The constant $0$ is 
 used to separate `positive' indexes - those satisfying $0\leq i$ - from the remaining `negative' ones.

Examples of index theories are $TO$ itself, integer difference logic \IDL, integer linear arithmetic \LIA, and  real linear arithmetics \LRA. In order to match the  requirements of Definition~\ref{def:index}, one need however to make a careful choice of the language (see~\cite{BGR14} for details): most importantly, notice that 
 integer (resp., real) division  by all positive 
integers should be added to the language of \LIA (resp. \LRA). For most applications, \IDL 
(
which is the theory of integer numbers with 0, ordering, successor and predecessor) 
is sufficient as in this theory one can model counters for scanning arrays.

Given an index theory $T_I$, we can now introduce our \emph{contiguous array theory with maxdiff} $\AXD(T_I)$ (parameterized by $T_I$) as follows.
We still have three sorts $\ARRAY, \ELEM, \INDEX$; the language includes the symbols of $T_I$, the read and write operations $wr, rd$, a binary function $\diff$
of type $\ARRAY \times \ARRAY \to \INDEX$, a unary function $\len{-}$ of type 
$\ARRAY\to \INDEX$,
as well as  constant $\bot, \el$ of sort \ELEM. 
The constant $\bot$ models an undefined value; the term $\diff(x,y)$ returns the maximum index where $x$ and $y$ differ and returns 0 if $x$ and $y$ are equal.~\footnote{Notice that it might well be the case that 
$\diff(x,y)=0$ for different $x,y$, but in that case $0$ is the only index where 
$x,y$ differ. } The term $\len{a}$ indicates the \emph{(strong) length} of $a$, meaning that $a$ is allocated in the interval $[0, \len{a}]$ and undefined outside. 
Formally, the axioms of $\AXDTI$
 include, besides the axioms of $T_I$, the following ones:  
\begin{equation}
\label{ax1}
\forall y,i,e,\ \len{wr(y,i,e)}=\len{y}
\end{equation}
\begin{equation}
\label{ax2}
\forall y,i,\ wr(y,i,\bot)=y 
\end{equation}
\begin{equation}
\label{ax3}
\forall y,i,e,\ (e\neq\bot \land \ 0\le i\le \len{y}) \rightarrow rd(wr(y,i,e),i)=e
\end{equation}
\begin{equation}
\label{ax4}
\forall y,i,j,e,\ i\neq j\rightarrow  rd(wr(y,i,e),j)=rd(y,j)
\end{equation}
\begin{equation}
\label{ax5}
\forall y,i,\ rd(y,i)\neq \bot \leftrightarrow 0\le i\le \len{y}
\end{equation}
\begin{equation}
\label{ax6}
\forall y,\ \len{y}\ge 0
\end{equation}
\begin{equation}
\label{ax7}
\forall y,\diff(y,y)=0
\end{equation}
\begin{equation}
\label{ax8}
\forall x,y,\ x\neq y \rightarrow rd(x,\diff(x,y))\neq rd(y,\diff(x,y)). 
\end{equation} 
\begin{equation}
\label{ax9} 
\forall x,y,i,\ \diff(x,y)<i \rightarrow rd(x,i)=rd(y,i).
\end{equation}
\begin{equation}
\label{ax12} 
\bot \neq \el.
\end{equation}

Axiom~\eqref{ax12} prevents the \ELEM\ sort to contain just $\bot$ (thus trivializing a model).
Since an array $a$ is fully allocated only in the interval $[0,\len{a}]$, any reading or writing attempt outside that interval should produce some runtime error; similarly, it should be impossible to overwrite $\bot$ inside that interval. In our declarative context, there is nothing like a `runtime error', so we assume that such illegal operations simply do not produce any effect. However, when applying the theory to code annotations, the verification conditions should include that no memory violation like the above ones occur (that is, when, e.g., a term like $rd(b,i)$ occurs, it should be accompanied by the proviso annotation $0\leq i\leq \len{a}$, etc.).

As we shall see the above theory enjoys amalgamation (i.e., plain quantifier-free interpolation) but not strong amalgamation (i.e., it lacks the general quantifier-free interpolation). To restore it, it is sufficient to add some (even limited) support for constant arrays: we 
call the related theory 
$\estesa$. The extension is interesting by itself, because it  increases the expressivity of the language: in $\estesa$, applying the $wr$ operation to terms $\Const(i)$,  one can encode  all finite lists (if $T_I$ has a reduct to \IDL).
Formally, 
$\estesa$ has an additional unary function $\Const:\INDEX\to \ARRAY$, constrained by the following axioms:
\begin{equation}
\label{ax10} 
\forall i, \len{\Const(i)} = \max(i,0).
\end{equation}
\begin{equation}
\label{ax11} 
\forall i,j,~ (0\leq j \wedge j \leq \len{\Const(i)} \to rd(\Const(i),j)=el).
\end{equation}
(we assume without loss of generality that $\max$ is a symbol of $T_I$ - in fact it is definable in it). 

The following easy facts will be often used in our proofs:

\begin{lemma}\label{lem:easy}
 The following \formulae are $\AXDTI$-valid
 \begin{equation}\label{eq:diffmax}
 \len{a} \neq \len{b} \to \diff(a,b)= \max(\len{a}, \len{b})~~
 \end{equation}
 \begin{equation}\label{eq:triangular}
  \max(\diff(a,b), \diff(b,c))\geq \diff(a,c)~~.
  \end{equation}
\end{lemma}

The
next lemma 
follows
 from the 
axioms of \AXDTI: 

\begin{lemma}\label{lem:univinst}
 An atom like $a=b$ is equivalent (modulo \AXDTI) to 
 \begin{equation}\label{eq:eleq}
  \diff(a,b)=0 \wedge rd(a,0)=rd(b,0)~.
 \end{equation}
An atom like $a=wr(b,i,e)$ is equivalent (modulo \AXD) to the conjunction of the following formulae
\begin{equation}
\begin{aligned}
    (e\neq \bot \land 0 \le i\le |b|)\rightarrow rd(a,i)=e\\
    (i<0 \lor i> |b| \lor e=\bot)\rightarrow rd(a,i)=rd(b,i)\\
    \forall h\ (h\neq i\rightarrow rd(a,h)=rd(b,h)).
\end{aligned}
\label{wrt3}
\end{equation}
An atom of the kind $\len{a}=i$ is equivalent to:  
\begin{equation}
\label{lunghequivt3}
    i\ge 0~\wedge ~
    \forall h\ (rd(a,h)\neq \bot \leftrightarrow 0\le h \le i).
\end{equation}
\end{lemma}

\begin{lemma}\label{lem:constinst}
 An atom like $\Const(i)=a$ is equivalent (modulo \estesa) to 
 \begin{equation}\label{constequivt}
 \len{a}=i \wedge \forall h\ (0\leq h\leq i \to rd(a,h)=el)~~.
 \end{equation}

\end{lemma}

Similarly to~\cite{TW16} and ~\cite{ourfossacs}, we now introduce iterated $\diff$ operations, that will be used in our interpolation algorithm. 
In fact, in addition to $\diff:=\diff_1$ we need 
an operator $\diff_2$ that returns the last-but-one index where $a,b$ differ (0 if $a,b$ differ in at most one index), 
an operator $\diff_3$ 
that returns the last-but-two index where $a,b$ differ (0 is they differ in at most two indexes), etc. Our language is already sufficiently expressive for that. 
Indeed, given array variables $a,b$, we define by mutual recursion the sequence of array terms $b_1, b_2, \dots$ and of index terms $\diff_1(a,b), \diff_2(a,b), \dots$:  
\begin{eqnarray*}
b_1~:=~b;\; ~~~
\diff_1(a,b):=~  \diff(a,b_1); \\
b_{k+1}~:=~wr(b_k, \diff_k(a,b), rd(a,\diff_k(a,b))); \\
\diff_{k+1}(a,b)~:=~\diff(a, b_{k+1}) ;
\end{eqnarray*}

A useful fact is that 
formulae like $\bigwedge_{j<l}\diff_j(a,b)=k_j$ can be eliminated in favor of universal clauses in a language whose only symbol  for array variables is $rd$. In detail:

\begin{lemma}\label{lem:elim}
 A formula like 
 \begin{equation}\label{eq:iterated_diff}
  \diff_1(a,b)=k_1 \wedge \cdots \cdots \wedge \diff_l(a,b)=k_l
 \end{equation}
is equivalent modulo $\AXDTI$ to the conjunction of the following 
seven formulae:
\begin{equation}
\begin{aligned}
\label{diffiterate}
    k_1\ge k_2 \land ... \land k_{l-1}\ge k_l \land k_l\ge 0\\
    \bigwedge_{j<l}(k_j>k_{j+1}\rightarrow rd(a,k_j)\neq rd(b,k_j))\\
    \bigwedge_{j<l}(\len{a}=\len{b} \land k_j=k_{j+1})\rightarrow k_j=0\\
    \bigwedge_{j\le l}(rd(a,k_j)= rd(b,k_j)\rightarrow k_j=0)\\
    \forall h\ (h>k_l \rightarrow rd(a,h)=rd(b,h)\lor h=k_1\lor ... \lor h=k_{l-1})\\
    \len{a} > \len{b} \rightarrow (k_1 = k_l \land k_l = \len{a})\\
    \len{b} > \len{a}\rightarrow (k_1 = k_l \land k_l = \len{b}).
\end{aligned}
\end{equation}
\end{lemma}

\section{Embeddings}\label{sec:embeddings}

In this section we present some useful facts about embeddings that will be crucial in the proofs throughout the paper. 

We first introduce the third array theory $\AXEXTTI$, 
which is weaker than \AXDTI, lacks the \diff\ symbol and axiom~\eqref{ax9} is replaced by the following extensionality axiom:
\begin{equation}\label{eq:ext}
\forall x,y,\ x\neq y \rightarrow (\exists i,\ rd(x,i)\neq rd(y,i)).
\end{equation}
Notice that 
$\AXEXTTI\subseteq \AXDTI \subseteq \estesa$ (the inclusion holds both for signatures and for axioms). To simplify the statements of some lemmas below, let us also introduce the theory \extestesa: this theory is obtained from \AXEXTTI by adding the function symbol $\Const$ to the signature and the sentences \eqref{ax10},\eqref{ax11} to the axioms.

 We
now discuss the class of models of $\AXEXTTI$ and we 
clarify the important features of 
embeddings between such models. 
A model $\cM$ of \AXEXTTI  is \emph{functional} when the following conditions are satisfied: 
\begin{compactenum}
 \item[{\rm (i)}]
$\ARRAY^\cM$ is a subset of the set of all positive-support functions from $\INDEX^\cM$ to $\ELEM^\cM$
(a function $a$ is \emph{positive-support} iff there exists an index $\len{a}$ such that $\len{a}\geq 0$ and,  
for every $j$,
$a(j)\neq\bot$ iff  $j\in [0,\len{a}]$); 
 \item[{\rm (ii)}] $rd$ is function application;
 \item[{\rm (iii)}] $wr$ is
the point-wise update operation inside the interval $[0,\len{a}]$ (i.e., function 
$wr(a,i,e)$ returns the same values as function $a$,
except at the index $i$ and only in case $i\in [0,\len{a}]$: in this case  
it returns the
element $e$); 
 \item[{\rm (iv)}] if $\cM$ is also a model of \extestesa, then the set $\ARRAY^\cM$ contains the positive-support functions 
 with value $\el^{\cM}$ inside their support. 
\end{compactenum}
Because of the extensionality axiom~\eqref{eq:ext},
it can be shown
that every
model of \AXEXTTI or of \extestesa
is \emph{isomorphic  to a functional one}. 
For an array $a\in \INDEX^\cM$ in a functional model $\cM$ and for $i\in \INDEX^\cM$,  since $a$ is a function, we interchangeably use the notations $a(i)$ and $rd(a,i)$.

Let $a,b$ be
elements of $\ARRAY^\cM$ in a model $\cM$.  We say that
\emph{$a$ and $b$ are cardinality equivalent}
iff $\len{a}=\len{b}$ and $\{i\in \INDEX^\cM \mid \cM\models rd(a,i)\neq
rd(b,i)\}$ is finite.  
This relation in $\cM$ is 
an equivalence, 
that we denote as $\sim_\cM$ or simply as $\sim$.
We also write $\cM\models a\sim b$ to say that $a\sim_\cM b$ holds.

\begin{lemma}\label{lem:dependency}
  Let $\cN$, $\cM$ be models of  \AXEXTTI 
  such that $\cM$ is a
  substructure of $\cN$.  For every $a,b\in\ARRAY^\cM$, we have that
  \begin{eqnarray*}  
    \cM\models a\sim b & \mbox{ {\rm iff} } &
    \cN\models a\sim b.
  \end{eqnarray*}
\end{lemma}

In a functional model $\cM$ of $\AXEXTTI$,
we say that $\diff(a,b)$ is \emph{defined} iff there is a maximum index where $a,b$ differ (or if $a=b$). If in the model $\cM$  the index sort $\INDEX$ is interpreted as  the set of the integers,
with standard ordering, then for any two positive-support functions $a,b$, we have that $\diff(a,b)$ is defined. However, this will not be the case if the index sort $\INDEX$ is interpreted, e.g., in some non-standard model  of the integers. We must take into considerations these models too, since 
we want to prove amalgamation. For this purpose, we need to build amalgams 
for \emph{all} models of the theory (only in that case in fact, amalgamation turns out to be equivalent to quantifier-free interpolation). Thus, we are forced to take into consideration below also phenomena that might arise only in non-standard models.

An embedding $\mu: \cM\longrightarrow \cN$ between $\AXEXTTI$-models (or of \extestesa-models) is said to be $\diff$-faithful iff, whenever $\diff(a,b)$ is defined, so is
$\diff(\mu(a),\mu(b))$ and it is equal to $\mu(\diff(a,b))$.
Since there might not be a maximum index where $a,b$ differ,  in principle it is not always possible to expand a functional model of $\AXEXTTI$ to a functional model of $\AXDTI$, if the set of indexes remains unchanged.
Indeed, in order to do that in a $\diff$-faithful way, 
one needs to explicitly add to $\INDEX^{\cM}$ 
new indexes including at least  
 the ones representing
the missing maximum indexes 
where two given array differ. 
This idea leads to Theorem~\ref{thm:extension} below, which is the main result of the current section.

\begin{theorem}\label{thm:extension} For every index theory $T_I$,
 every  model $\cM$ of $\AXEXTTI$ (resp. of \extestesa) has a $\diff$-faithful embedding into a model of $\AXD(T_I)$ (resp. of \estesa).
\end{theorem}
\begin{proof}
 It is sufficient to well-order the pairs $a,b\in \INDEX^\cM$ such that $\diff(a,b)$ is not defined in $\cM$, apply to each pair the contruction of the next lemma (taking unions at limit ordinals), and then repeat the
 whole
 construction $\omega$ times, taking union again.
\end{proof}

\begin{lemma}
\label{immersionediff} Let $\cM$ be a model of $\AXEXTTI$ (resp. of \extestesa) and let
$a,b\in \ARRAY^\cM$ be such that $\diff^{\cM}(a,b)$ is not defined. 
Then there are a model $\cN$ of $\AXEXTTI$ (resp. of \extestesa) and a \diff-faithful embedding
$\mu:\cM\rightarrow \cN$ such that $\diff^\cN(a,b)$ is defined. 
\end{lemma}

\begin{proof} 
We can assume that 
$\ELEM^\cM$ 
has at least an element
 $e$, different from $\bot^\cM, \el^\cM$ (for details, see Lemma~\ref{aggiungoelem} in the appendix). 
 We must have $\len{a}=\len{b}$, otherwise $\diff(a,b)$ is defined and it is $\max(\len{a},\len{b})$ according to Lemma~\ref{lem:easy}.

 Let $I=\{i\in \INDEX^\cM\mid a(i)\neq b(i)\}$ be the set of indices without maximum element (hence infinite) where they differ. 
 Let $\downarrow I := \{j\in \INDEX^\cM \mid \exists i\in I,\ j\le i\}\supseteq I$.
 The condition
 $$
 (+)~~ ``\exists i\in I\,\forall j\in I\, ( j\geq i \to x(j)=\el)"
 $$
 cannot be satisfied  both for $x=a$ and $x=b$ (see the appendix). 
  In case one of them satisfies it, we assume it is $b$.
  
Let
$\Delta$ be the Robinson diagram of the $T_I$-reduct of $\cM$ and let $k_0$
be  a new constant; let us introduce the set
$$\Delta' := \Delta\cup\{ i<k_0 \mid i \in\,\downarrow I\} \cup
\{k_0<i \mid i \in \INDEX^\cM\setminus \downarrow I\}.
$$
By compactness theorem and since $I$ is infinite, the set 
 $\Delta'$ turns out to be consistent (see again the appendix for details). 

By Robinson Diagram Lemma, there exists a model $\cA$ of $T_I$ extending the 
$T_I$-reduct of  $\cM$; such $\cA$ contains in its support an element  $k_0$ such that
$$
\forall i\in\, \downarrow I,\ i<k_0,
$$ $$
\forall i\in \INDEX^\cM\setminus \downarrow I,\ k_0<i.
$$

We now take 
 $\ELEM^\cN=\ELEM^\cM$, $\INDEX^\cN=\INDEX^\cA$; we let also  $\ARRAY^\cN$
 to be the set of all positive-support functions from
  $\INDEX^\cN$ into $\ELEM^\cN$ (notice that this $\cN$ is trivially also a model of \extestesa).
  We observe  that $k_0< \len{a}^\cM$ and recall that $\len{a}^\cM=\len{b }^\cM$.
  
Let us now define the embedding 
 $\mu:\cM\rightarrow \cN$; at the level of the sorts \INDEX and \ELEM, we use inclusions. For the \ARRAY sort, we need to specify the value 
  $\mu(c)(k)$ for $c\in \ARRAY^\cM$ and $k\in \INDEX^\cN\setminus \INDEX^\cM$
  (for the other indices we keep the old $\cM$-value to preserve the read operation). Our definition for $\mu$ must preserve the maxdiff index (whenever already defined in $\cM$) and must guarantee that  $\diff^\cN(\mu(a),\mu(b))=k_0$ (by construction, we have $k_0>0$). 
  For a generic array 
   $c\in \ARRAY^\cM$, we operate as follows:
\begin{compactenum}
    \item if $\len{c}^{\cM}<k_0$ we put $\mu(c)(k_0)=\bot^\cM$, otherwise:
    \item if the condition $(\star)$ below holds, we put $\mu(c)(k_0)=e$,
    \item if such condition does not hold, we put $\mu(c)(k_0)=el^\cM$.
\end{compactenum}
The condition  $(\star)$ is specified as follows:
\begin{description}
 \item[$(\star)$] ~~
there is $i\in I$ such that for all $j\in I, j\geq i$ we have $c(j)=a(j)$.
\end{description}
\noindent
For all the remaining indexes $k\in \INDEX^\cN\setminus (\INDEX^\cM\cup \{k_0\})$: 
\begin{equation}
\label{indicinonk0}
    \mu(c)(k)=\begin{cases} 
\bot^\cM, & \mbox{if $k \notin [0, \len{c}^\cM]$}\\
\el^\cM, & \mbox{if $k \in [0, \len{c}^\cM]$}.
    \end{cases}
\end{equation}
Notice that   we have $\mu(a)(k_0)=e\neq \el^\cM=\mu(b)(k_0)$ (because $I$ is infinite and does not have maximum, hence condition  $(\star)$ holds for $a$ but not for $b$). 
   In addition:
\begin{compactenum}[$\bullet$]
    \item for all $i\in\INDEX^\cM$ such that $k_0<i$, 
    we have $i\notin\, \downarrow I$ (according to the construction of $k_0$) and
    consequently $i\notin\, I$, that is $a(i)=b(i)$;
    \item for all $i\in\INDEX^\cN\setminus (\INDEX^\cM\cup\{k_0\})$ such that $k_0<i$, since we have  $\len{a}^\cM=\len{b}^\cM$, we get
    $$\mu(a)(i)=\bot^\cM~{\rm iff}~ \mu(b)(i)=\bot^\cM,$$
    $$\mu(a)(i)=\el^\cM ~{\rm iff}~ \mu(b)(i)=\el^\cM.$$
\end{compactenum}
Hence, we can conclude that  $\diff^\cN(\mu(a),\mu(b))$ is defined and equal to $k_0$.

We only need to check that our $\mu$ preserves $rd, \len{-}, wr$, constant arrays and \diff\ 
(whenever defined). For preservation of constant arrays, we need to prove only that $\mu(\Const(i))(k_0)=\el^\cM$ in case $k_0<i$: this is clear, because $a$ does not satisfy (+), hence  $(\star)$ does not hold for $\Const(i)$. The other cases are proved in the appendix.
\end{proof}

\section{Strong Amalgamation for \estesa}\label{sec:strongamalg}

In this section, we prove that the most expressive theory of the paper \estesa has strong amalgamation. However, we also show that this is not the case for \AXDTI (even if it is amalgamable).
We recall that strong amalgamation holds for models of $T_I$  (see Definition~\ref{def:index}): this observation is crucial for the  following. %

Strong amalgamation of \estesa will be proved in two steps. First, we provide the amalgam construction for \estesa, where we also notice that the same arguments can be used to prove that \AXDTI has amalgamation too. Then, after exhibiting a counterexample showing that the strong amalgamation fails for \AXDTI, we check that the amalgam construction for \estesa satisfies the condition for being a $\estesa$-strong amalgam.

\subsection{Amalgam constructions}
Let 
 $\cM_1$ and $\cM_2$ be two models of $\estesa$  (resp. of $\AXDTI$); we want to amalgamate them over their common substructure $\cA$ and let $f_i$ be the embedding of $\cA$ into $\cM_i$
 (we assume that $f_i$ is just inclusion for the \INDEX and \ELEM components).
 We can assume w.l.o.g. 
 that our models are all functional and, by applying renaming, that 
$$
(\INDEX^{\cM_1}\setminus \INDEX^{\cA})\ \cap\ (\INDEX^{\cM_2}\setminus \INDEX^{\cA})\ =\ \emptyset 
$$ $$
(\ELEM^{\cM_1}\setminus \ELEM^{\cA})\ \cap\ (\ELEM^{\cM_2}\setminus \ELEM^{\cA})\ =\ \emptyset . 
$$ 

\noindent 
We build the amalgamated model in two steps.
We first embed  $\cM_1$ and $\cM_2$, via  the embeddings $\mu_i$ (i=1,2), into a model $\cM$ of \extestesa (resp., of $\AXEXTTI$) in a \diff-faithful way. Then $\cM$ is embedded, via another \diff-faithful embedding  $\mu'$ into a model $\hat{\cM}$ of  \estesa (resp., of $\AXDTI$): $\mu'$ is guaranteed
to exist
by Theorem~\ref{thm:extension}. 
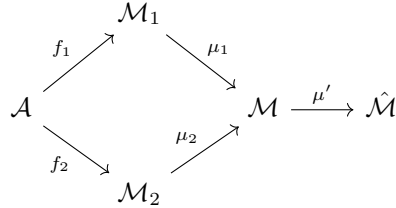
$$
\begin{tikzcd}
& \cM_1 \ar[dr,"\mu_1"] 
&
&
& \\
\cA \ar[ur, "f_1"] \ar[dr, "f_2"']
&
& \cM \ar[r,"\mu'"]
& \hat{\cM}
&\\
& \cM_2 \ar[ur,"\mu_2"]
&
\end{tikzcd} 
$$

\subsection*{Construction of $\mu_i$}
We build the model $\cM$ 
and the two \diff-faithful embeddings $\mu_i:\cM_i\to \cM$ such that $  \mu_{1} \circ f_1 = \mu_{2} \circ  f_2 $. 

\noindent We let $\INDEX^{\cM}$ be a strong amalgam of 
$\INDEX^{\cM_1}$ and $\INDEX^{\cM_2}$ ($T_I$
enjoys strong amalgamation), whereas we let
$
\ELEM^{\cM}\ =\ \ELEM^{\cM_1}\cup \ELEM^{\cM_2} 
$. 
Let 
$\ARRAY^\cM$ be the set of all positive-support functions from 
 $\INDEX^\cM$ into $\ELEM^\cM$.

\noindent 
The $\INDEX$ and $\ELEM$ components of the embeddings $\mu_i$ will be inclusions. The definition of the value of $\mu_i(a)(k)$, for $a\in \ARRAY^{\cM_{i}}$ and $k \in \INDEX^\cM$, is given by cases as follows:
\begin{compactenum}[$\bullet$]
    \item if $k\in \INDEX^{\cM_i}$, we put $\mu_{i}(a)(k)= a(k)$;
    \item if $k \in \INDEX^{\cM_{3-i}} \setminus \INDEX^\cA$:
    let 
     ($2\star$) be the relation\footnote{When we write  $k>\diff^{\cM_i}(b,f_i(c))$ 
    we mean in fact that
     $k>\mu_i(\diff^{\cM_i}(b,f_i(c)))$ (this relation is meant  to hold in $\cM$).
     Our simplified notation is justified by the fact that 
      $\mu_i$ is inclusion for \INDEX sort.}
    \begin{eqnarray*} ''{\rm there~exist}~ c\in \ARRAY^\cA, b\in \ARRAY^{\cM_i}\ \\ {\rm s.t.} \ b\sim^{\cM_i} a,\ k>\diff^{\cM_i}(b,f_i(c))''
    ,\end{eqnarray*}  
    $$\small
    \mu_{i}(a)(k)=\begin{cases} 
f_{3-i}(c)(k), & \mbox{if ($2\star$) holds}\\ 
\bot^\cM, & \mbox{if ($2\star$) does not hold \& $k \notin [0, \len{a}^{\cM_i}]$}\\
el^\cM, & \mbox{if ($2\star$) does not hold \& $k \in [0, \len{a}^{\cM_i}]$}
    \end{cases}
    $$\small
    \item if $k \notin \INDEX^{\cM_{i}}\cup \INDEX^{\cM_{3-i}}$, we put
   $$ \mu_{i}(a)(k)=\begin{cases} 
\bot^\cM, & \mbox{if $k \notin [0, |a|^{\cM_i}]$}\\
el^\cM, & \mbox{if $k \in [0, \len{a}^{\cM_i}]$}.
    \end{cases}$$
\end{compactenum}

We need to prove that the functions 
 $\mu_i$: (i) are well-defined, (ii) are injective, (iii) preserve $\len{-}$,
 (iv) preserve 
 $rd$ and $wr$, (v) preserve  $\diff$, (vi) satisfy the condition $ \mu_{1} \circ f_1=  \mu_{2} \circ f_2 $, (vii) preserve constant arrays (this point is not needed for $\AXDTI$). 
 These checks 
 require a careful analysis of all the cases: all the details can be found in the appendix.

 As a consequence of these constructions, we get the following result for both \estesa and $\AXDTI$.
 \begin{theorem}\label{thm:axd_amalg}
 $\AXDTI$ and \estesa enjoy the amalgamation property.
\end{theorem}

\subsection{The \estesa-amalgam is strong}

We now prove the main result of the section, i.e., strong amalgamation for \estesa.  
Unfortunately, this property does not hold for \AXDTI: as it is shown in the example below, we need constant arrays in the language 
(recall that strong amalgamation is equivalent to general quantifier-free interpolation, see Theorem~\ref{thm:interpolation-amalgamation}(ii)). 

\begin{example}\label{example1}
 Consider the following two formulae (where $P$ is a free predicate symbol):
 \begin{eqnarray*}
  (A)~~~& \len{a}=0 \wedge rd(a,0)= e \wedge P(a)~~ \\
  (B)~~~& \len{b}=0 \wedge rd(b,0)= e \wedge \neg P(b). 
 \end{eqnarray*}
The conjunction $(A)\wedge (B)$ is inconsistent because $a$ and $b$ are in fact the same array (because of Axioms~\ref{ax5} and~\ref{ax8}). 
  However, the only common variable is $e$; to get the interpolant, we can use $P(wr(\Const(0),0,e))$, 
  but then it is clear that the language lacking constant arrays does not suffice.
  \eop
 
  \end{example}

\begin{theorem}\label{thm:axd_strong_amalg} 
 $\estesa$ has strong amalgamation. 
\end{theorem}

\begin{proof}
We keep the same notation and construction as in the proof of Theorem~\ref{thm:axd_amalg}. However, it can be proved (for details, see
 Lemma~\ref{lemmaka} in the appendix) 
  that all arrays 
 $a\in \ARRAY^{\cM_i}$  (for $i=1,2$) whose length belongs to $\INDEX^\cA$ can be assumed to be s.t. one of the following two conditions are satisfied:
\begin{enumerate}
    \item there exists $c\in \ARRAY^\cA$ with $f_i(c)\sim^{\cM_i}a$;
    \item there exists $k_a\in \INDEX^{\cM_i}\setminus \INDEX^\cA$ 
    such that $a(k_a)$ is an element from $\ELEM^{\cM_i}\setminus \ELEM^\cA$ different from all the $f_i(c)(k_a)$, varying
     $c\in \ARRAY^\cA$. 
\end{enumerate}

Let 
$a_i\in \ARRAY^{\cM_i}$ ($i=1,2$) be s. t. $\forall k\in \INDEX^\cM$ we have
\begin{equation}
\label{indiciamalgama3}
  \mu_1(a_1)(k)=\mu_2(a_2)(k).  
\end{equation}
Notice that, since $T_I$ has the strong amalgamation property and the $\mu_i$ preserve length, this can only happen if $\len{a_1}=\len{a_2}$ belongs to $\INDEX^\cA$. We look for some
$c\in \ARRAY^{\cA}$ such that  $a_1=f_1(c)$; since $\mu_2$ is injective this would entail  $a_2=f_2(c)$ because
$
\mu_2(a_2)=\mu_1(a_1)=\mu_1(f_1(c))=\mu_2(f_2(c)),
$   
implying that $\hat{\cM}$ is a strong amalgam, as requested. 

We separate two cases: (i) one of the arrays $a_1, a_2$ satisfy 
condition 2; (ii) both arrays $a_1, a_2$ satisfy 
condition 1. 
\begin{compactenum}
    \item [(i)]
    This case is impossible: this is proved in the appendix, by exploiting condition 2 and the definition of $\mu_2$ via rule ($2\star$). 
    \item [(ii)] Hence, we have 
    $\mu_1(a_1)=\mu_2(a_2)$ only when both $a_1,a_2$ satisfy condition 1. Let
    $c_i\in \ARRAY^\cA$ ($i=1,2$) be the arrays s.t. $f_i(c_i)\sim^{\cM_i}a_i$. In the appendix, we prove how to build, from $c_2$, a
     $c\in \ARRAY^\cA$ such that $f_1(c)=a_1$. 
\end{compactenum}
\end{proof}

Strong amalgamation corresponds to general quantifier-free interpolation 
(Theorem~\ref{thm:interpolation-amalgamation}), 
hence we obtain that: 

\begin{corollary}
\label{risultatot3} The theory $\estesa$ has the general quantifier-free interpolation property.
\end{corollary}

\section{Satisfiability}\label{sec:sat}

We now address the problem of checking satisfiability  of quantifier-free \formulae in our theories. Decidability of the $SMT(\AXDTI)$- and $SMT(\estesa)$-problems, at least in the relevant case where $T_I$ is any fragment of Presburger arithmetics, can be solved by reduction to the satisfiability problem for the so-called `array-property fragment' of ~\cite{BMS06}: the reduction 
can be obtained by eliminating the $wr, \len{-}, \diff$ and $\Const$ symbols in favor of universally quantified \formulae belonging to that fragment (see Lemmas~\ref{lem:univinst},\ref{lem:constinst},\ref{lem:elim}). However, we now supply a decision procedure for quantifier-free \formulae, since this will be useful for the interpolation algorithm in Section~\ref{sec:algo}. 


A \emph{flat} literal $L$ is a formula of the kind $x_0=f(x_1,\dots, x_n)$  or $x_1\neq x_2$ or $R(x_1, \dots, x_n)$ or $\neg R(x_1, \dots, x_n)$, where the $x_i$ are variables, $R$ is a relation symbol, 
and $f$ is a function symbol.
If $\cI$ is a  set of $T_I$-terms, an \emph{$\cI$-instance} of a universal formula of the kind  $\forall i\,\phi$ is a formula of the kind $\phi(t/i)$ for some $t\in \cI$.

\begin{definition}\label{def:separated}
A pair of sets of quantifier-free \AXDTI-\formulae $\Phi=(\Phi_1, \Phi_2)$ is a
 \emph{separated pair} 
iff
\begin{compactenum}
\item[{\rm (1)}] $\Phi_1$ contains  equalities of the form $\len{a}=i,\diff_k(a,b)=i$ and $a=wr(b,i,e)$; moreover if it contains the equality $\diff_k(a,b)=i$, it must also contain an equality of the form $\diff_l(a,b)=j$ for every $l<k$; finally, if $\Phi_1\cup \Phi_2$ contains occurrences of an array variable $a$, $\Phi_1$  must contain also an equality of the form $\len{a}=i$;
 \item[{\rm (2)}] $\Phi_2$ contains 
 Boolean combinations of $T_I$-atoms  and of atoms of the forms:
 \begin{equation}\label{phi2}
 rd(a,i)= rd(b,j), ~~
 rd(a,i)= e,~~
 e_1=e_2,
 \end{equation}
 where $a,b,i,j,e,e_1,e_2$ are variables or constants of the appropriate sorts.
\end{compactenum}
$\Phi$ is said to be finite iff $\Phi_1$ and $\Phi_2$ are both finite. 
\end{definition}

Notably, in a separated pair $\Phi=(\Phi_1, \Phi_2)$, 
if we introduce a unary function symbol $f_a$ for every array variable $a$ and
rewrite $rd(a,i)$ as $f_a(i)$, it turns out that  \emph{ the \formulae from $\Phi_2$ can be seen as \formulae of the combined theory $T_I\cup\EUF$}. $T_I\cup \EUF$  
is decidable in its quantifier-free fragment and 
admits quantifier-free interpolation because $T_I$ is an index theory (see Nelson-Oppen results~~\cite{NO79} and
Theorems~\ref{thm:interpolation-amalgamation},\ref{thm:ti+euf}): we  adopt a hierarchical approach (similarly to~\cite{SS08,SS18}) and \emph{we rely  on satisfiability and interpolation algorithms for such a  theory as  black boxes}.

\begin{definition}\label{def:instsep}
 Let $\cI$ be a set of $T_I$-terms and let $\Phi=(\Phi_1, \Phi_2)$ be a separated pair;
 $\Phi(\cI)=(\Phi_1(\cI), \Phi_2(\cI))$ is the smallest separated pair satisfying the following conditions:
 \begin{compactenum}
  \item[-] $\Phi_1(\cI)$ is equal to $\Phi_1$ and $\Phi_2(\cI)$ contains $\Phi_2$;
  \item[-] if $\Phi_1$ contains the atom $a=wr(b,i,e)$ then $\Phi_2(\cI)$ contains \emph{all the $\cI$-instances of the 
  formulae~\eqref{wrt3}} (with the terms 
  $\len{a}, \len{b}$ replaced by the index constants $i,j$ such that $\len{a}=i, \len{b}=j\in \Phi_1$, respectively);
  \item[-] if $\Phi_1$ contains the atom $\len{a}=i$, then 
  $\Phi_2(\cI)$ contains all the 
  $\cI$-instances of  the  formulae~\eqref{lunghequivt3};
  \item[-] if $\Phi_1$ contains the conjunction $\bigwedge_{i=1}^l \diff_i(a,b)= k_l$, then 
  $\Phi_2(\cI)$ contains   the  formulae~\eqref{diffiterate} (with the terms 
  $\len{a}, \len{b}$ replaced by the index constants $i,j$ such that $\len{a}=i, \len{b}=j\in \Phi_1$, respectively).
 \end{compactenum}
A separated pair $\Phi$ is \emph{0-instantiated} iff $\Phi=\Phi(\cI)$, where 
$\cI$ is the set of index variables occurring in $\cI$.
\end{definition}

 We say that a separated pair $\Phi=(\Phi_1, \Phi_2)$ is 
 is \AXDTI-satisfiable iff so it is the formula $\bigwedge \Phi_1\wedge \bigwedge \Phi_2$.

\begin{lemma}\label{lem:sat1}
 Let $\phi$ be a quantifier-free formula; then it is possible to compute 
 in linear time
 a finite 
 separation pair $\Phi=(\Phi_1, \Phi_2)$
 such that $\phi$ is \AXDTI-satisfiable iff so is  $ \Phi$.
\end{lemma}

\begin{lemma}\label{lem:sat2} 
 The following conditions are equivalent for a finite 0-instantiated separation pair $\Phi=(\Phi_1, \Phi_2)$:  
 \begin{compactenum}
 \item[{\rm (i)}] $\Phi$ is \AXDTI-satisfiable;
 \item[{\rm (ii)}] $\bigwedge \Phi_2$ is $T_I\cup \EUF$-satisfiable.
 \end{compactenum}
\end{lemma} 

From Lemmas~\ref{lem:sat1} and~\ref{lem:sat2}, we get the following result: 

\begin{theorem}\label{thm:sat}
 The  $SMT(\AXDTI)$ problem is decidable for every
 index theory $T_I$.
\end{theorem}  

Regarding complexity for the $SMT(\AXDTI)$ problem, 
notice that the satisfiability of the quantifier-free fragment of common index theories (like \IDL, \LIA, \LRA) is decidable in NP; hence, 
for such index theories, 
an NP bound for our $SMT(\AXDTI)$-problems is
easily obtained,  
because 0-instantiation is clearly finite and  polynomial (all strings of universal quantifiers to be instantiated have length one).

The above decidability and complexity results apply also to \estesa: one only simply has to allow  the $\Phi_1$-component of a separation pair to contain also atoms of the form $\Const(i)=a$ and Definition~\ref{def:instsep} to require
that $\Phi_2(\cI)$ contains all the $\cI$-instances of the 
  formulae~\eqref{constequivt} in case $Const(i)=a\in \Phi_1$.

\section{The interpolation algorithm}\label{sec:algo}

Since amalgamation is equivalent to quantifier-free interpolation for universal theories such as \AXDTI and \estesa (thanks to Theorem~\ref{thm:interpolation-amalgamation}),  Theorem~\ref{thm:axd_amalg} 
guarantees that \AXDTI and \estesa admit 
quantifier-free interpolation. 
However, the proof of Theorem~\ref{thm:axd_amalg} is not constructive: hence, 
 in order to compute an interpolant  for  an 
unsatisfiable conjunction like $\psi(\ux,\uy)\wedge \phi(\uy,\uz)$,  one 
needs in principle to enumerate all quantifier-free \formulae $\theta(\uy)$ that    are  consequences of $\phi$ and are inconsistent with $\psi$. Since the quantifier-free fragments of \AXDTI and \estesa are decidable, this is an effective procedure and, 
considering the fact that interpolants of jointly unsatisfiable pairs of \formulae exist, it also terminates. However, 
this type of algorithm is not practical. In this section, we 
provide a better and more practical algorithm that relies on 
 a hierarchical reduction to $T_I\cup \EUF$. Our algorithm works for \AXDTI only; 
for \estesa, 
we make some comments in Section~\ref{sec:conclusions}.

Our problem is the following: given two quantifier-free formulae $A^0$ and $B^0$ such that $A^0\wedge B^0$ is not satisfiable (modulo \AXDTI), to compute a quantifier-free formula $C$ such that 
\begin{compactenum}[(i)]
\item
$\AXDTI \models A^0\to C$; 
\item $\AXDTI \models C\wedge B^0\to \bot$;
\item
$C$ contains only the variables (of sort \INDEX, \ARRAY, \ELEM) which occur both in $A^0$ and in $B^0$.
\end{compactenum}

Below, we 
work with ground formulae over signatures expanded with free constants instead of quantifier-free formulae. 
We use letters $A, B, \dots$ 
for finite sets of ground formulae; the logical reading of a set of
formulae is the conjunction of its elements.  
For a signature $\Sigma$ and a set $A$ of formulae, $\Sigma^A$ denotes the signature $\Sigma$
expanded with the free constants occurring in $A$.
Let $A$ and $B$ be two finite sets of ground formulae in the
signatures $\Sigma^A$ and $\Sigma^B$, resp., and $\Sigma^C :=
\Sigma^A \cap \Sigma^B$.  
We `color' a term, a literal, or a formula $\varphi$ by calling it:

\begin{compactenum}[$\bullet$]
\item {\em \abcommon} iff it is defined over $\Sigma^C$;
\item {\em \alocal} (resp. {\em \blocal}) if it is defined over
  $\Sigma^A$ (resp. $\Sigma^B$);
\item {\em $A$-strict} (resp. {\em $B$-strict}) iff it is \alocal
  (resp. \blocal) but not \abcommon;
\item {\em strict} if it is either \astrict or \bstrict.
\end{compactenum}
  
 


There are a number of manipulations that can be freely applied to a jointly unsatisfiable pair $A,B$ without compromising the possibility of extracting an interpolant out of them. A list of such manipulations (called
`metarules') is supplied in~\cite{lmcs_interp},~\cite{BGR14}. Here we need to introduce only  some of them:

\begin{compactenum}[(i)]
 \item we can add to $A$ an \alocal quantifier-free formula entailed by $A$ (similarly we can add to $B$ a \blocal quantifier-free formula entailed by $B$): the interpolant computed after such a transformation is trivially an interpolant for the original pair too;
 \item we can pick an \alocal term $t$ and a fresh constant $x$ (to be considered \astrict from now on)
 and add to $A$ the equality $x=t$: again, the interpolant computed after such a transformation is trivially an interpolant for the original pair too (the same observation extends to $B$);
 \item we can pick an \abcommon term $t$ and a fresh constant $x$ (to be considered \abcommon from now on)
 and add to both $A$ and $B$ the equality $x=t$: in this case, if $\theta$ is the interpolant computed after 
 such a transformation, then $\theta(t/x)$ is an interpolant for the original pair.
\end{compactenum}

We shall often apply the above metarules (i)-(ii)-(iii) in the sequel.

\subsection{The Algorithm}

\emph{We reduce the problem of finding an interpolant of an unsatifiable pair $(A^0, B^0)$ to an analogous polynomial size problem in the weaker theory $T_I\cup \EUF$. }

Our unsatisfiable pair $(A^0, B^0)$ needs to be \emph{preprocessed}. 
Using the procedure in the proof of Lemma~\ref{lem:sat1},
we can suppose that  both $A^0$ and $B^0$ are given in the form of finite  separated pairs.
In fact, the procedure of Lemma~\ref{lem:sat1} just introduces constants in order to explicitly name terms, so that it fits within the above explained remarks (see metarules (ii)-(iii)). The newly introduced constants are colored \astrict, \bstrict or \abcommon depending on the color of the terms they name.
Notice that because of this preprocessing, for every \astrict (resp. \bstrict, \abcommon) array constant $a$, in $A^0$ (resp. $B_0$, $A^0\cap B^0$) there is an atom of the kind $\len{a}=l_a$.

 To sum up, $A^0$ is of the form $\bigwedge A^0_1\wedge \bigwedge A^0_2$ and 
 $B^0$ is of the form $\bigwedge B^0_1\wedge \bigwedge B^0_2$,
 for separated pairs $(A^0_1,A^0_2)$ and $(B^0_1, B^0_2)$.

Our interpolation algorithm consists of \emph{three transformation steps} (all of them 
fit our metarules (i)-(ii)-(iii)).
We let $N_A$ (resp. $N_B$) be the number of \alocal (resp. \blocal) index constants occurring 
within a $wr$ symbol
in $A^0$ (resp. $B^0$); we let also  $N$ be equal to $1+max(N_A, N_B)$.

\vskip 2mm\noindent
\framebox{\textbf{Step 1}.} This transformation must be applied for every  pair of  
 distinct \abcommon \ARRAY-constants $c_1, c_2$.
 The transformation  picks
 fresh \INDEX constants $k_1, \dots, k_N$
 (to be colored \abcommon)
 and adds the atoms 
 $\diff_n(c_1,c_2)=k_n$ (for all $n=1, \dots ,N$) to both sets  $A_1$ and $B_1$.
This transformation fits metarule (iii).
\vskip 2mm\noindent
\framebox{\textbf{Step 2}.} We apply 0-instantiation, that is we replace $A$ with $A(\cI_A)$ and $B$ with $B(\cI_B)$, where $\cI_A$ is the set of \alocal index constants and $\cB$ is the set of \blocal index constants 
(see Definition~\ref{def:instsep}).
This transformation fits metarule (i).
\vskip 2mm\noindent
\framebox{\textbf{Step 3}.}
As proved in Theorem~\ref{thm:al} below, at this step
 $A_2\wedge B_2$ is $T_I\cup \EUF$-inconsistent; 
 since $T_I\cup \EUF$ has quantifier-free interpolation by Theorem~\ref{thm:ti+euf},
 we can compute an interpolant $\theta$ of the jointly unsatisfiable pair $A_2,B_2$. To get our desired \AXDTI-interpolant, we only have to replace back
  in it the fresh \abcommon constants introduced by our trasformations by the \abcommon terms they name.

\begin{example}\label{ex2}
We let 
\noindent\begin{eqnarray*}
 A^0 \equiv   \{ \diff(a_1,a_2)=j,\diff(a_1,c_1)=j_1,\diff(a_2,c_2)=j_2 \}
 \\
 B^0 \equiv   \{j<l,~j_1<l,~ j_2<l,~ rd(c_1,l)\neq rd(c_2,l) \}
\end{eqnarray*}
In the preprocessing step, we must add the atoms $\len{a_1}=l_{a_1}, \len{a_2}=l_{a_2}$ to $A^0_1$ and 
$\len{c_1}=l_{c_1}, \len{c_2}=l_{c_2}$ to both $A^0_1$ and $B^0_1$.
Since $N_A=N_B=0$, we have $N=1$; Step 1 adds the \abcommon atom $\diff(c_1,c_2)=k_1$. 
Step 2 makes the required 0-instantiations producing the 0-instantiated separated pair $(A, B)$. From such instantiations, we can conclude that $A_1$ entails $k_1\leq \max (j_1,j_2, j)$; this  is $T_I\cup \EUF$-inconsistent with $B_1$ 
as $B_1$ contains, in addition to $j<l,~j_1<l,~ j_2<l, rd(c_1,l)\neq rd(c_2,l)$,  also
$$
k_1 < l \to rd(c_1,l)= rd(c_2,l)
$$
by~\eqref{diffiterate}. Thus $k_1\leq \max (j_1,j_2, j)$ is a $T_I\cup \EUF$-interpolant. Using the recover instruction of meta-rule (iii), we get $\diff(c_1,c_2)\leq \max (j_1,j_2, j)$ as an \AXDTI-interpolant. 
\end{example}

\begin{example}\label{ex3}
 Let $A^0$ be
 \begin{eqnarray*}
 \{\diff(a,c_1)=i_1,\diff(b,c_2)=i_1,a_1=wr(b, i_1, e_1),\len{a}=k,~~\\ ~a=wr(a_1,i_3, e_3),
  \len{a_1}=k,~\len{b_1}=k,~\len{c_1}=k,~\len{c_2}=k \}~~~~~~~~~~~~~~~~~~~~~~~~~~~~~~~~~~
 \end{eqnarray*}
 and let $B^0$ be
$\{ rd(c_1,i_1)\neq rd(c_2,i_2), ~i_1<i_2,~ i_2<i_3,\\~i_3<k,~~\len{c_1}=k,~\len{c_2}=k\}$. 

We do not need any preprocessing here; since $N=3$, Step 1 adds the \abcommon atoms 
$$
\diff_1(c_1,c_2)=k_1,~\diff_2(c_1, c_2)=k_2,~\diff_3(c_1,c_2)=k_3~.
$$
Step 2 produces a separated pair $(A,B)$ such that 
 $A_2\wedge B_2$  is $T_I\cup
\EUF$-in\-con\-sistent (inconsistency can be tested via an SMT-solver like \textsc{z3}~\cite{z3} or \textsc{MathSat}~\cite{brutto08}).  
 The related $T_I\cup
\EUF$-interpolant (once $k_1, k_2$ and $k_3$ are replaced by 
   $\diff_1(c_1,c_2), \diff_2(c_1,c_2)$ and $\diff_3(c_1,c_2)$, respectively) gives our \AXDTI-interpolant.
 \end{example}

\begin{example}\label{ex4}
 This is the classical example (due to R. Jhala) showing that \AXEXT does not have quantifier-free interpolation (one needs \diff\ in the signature to recover it). Let $A^0$ be
 $
 \{c_1 =wr(c_2,i,e)\}
 $
 and $B^0$  be
 $
 \{i_1\neq i_2,~rd(c_1,i_1)\neq rd(c_2,i_1),~rd(c_1,i_2)\neq rd(c_2,i_2)\}. 
 $
 Preprocessing adds the \abcommon literals $\len{c_1}=l_{c_1}, \len{c_2}=l_{c_2}$ to both $A^0_1$ and $B^0_1$.
 Step 1 introduces the \abcommon atoms 
 $$
\diff_1(c_1,c_2)=k_1,~\diff_2(c_1, c_2)=k_2~.
$$
In  the 0-instantiated separated pair $(A,B)$ produced by Step 2, from~\eqref{wrt3} and~\eqref{diffiterate}, we realize that 
\begin{equation}\label{eq:ex4}
k_2=0~\wedge~ (k_1=k_2\vee rd(c_1,k_2)= rd(c_2,k_2))
\end{equation}
is $T_I\cup\EUF$-implied
by $A_1$ and $T_I\cup\EUF$-inconsistent with $B_1$. To get an \AXDTI-interpolant, it is 
enough to
replace $k_1, k_2$ respectively by $\diff_1(c_1,c_2),\diff_2(c_1, c_2)$
in~\eqref{eq:ex4}.
\end{example}

\begin{theorem}\label {thm:al}
 The above rule-based algorithm computes a quantifier-free interpolant for every 
 \AXDTI-mutually unsatifiable pair $A^0, B^0$ of quantifier-free formulae.
\end{theorem}

\begin{proof} 
We only need to prove that Step 3 really applies. 

Suppose not;
let $A=(A_1, A_2)$ and 
$B=(B_1, B_2)$ be the separated pairs obtained  
after  applications of Steps 1 and 2. If Step 3 does not apply, then $A_2\wedge B_2$ is $T_I\cup \EUF$-consistent.
  We claim that  $(A, B)$ is \AXDTI-consistent (contradicting 
  that $(A^0, B^0)\subseteq (A, B)$ was  
\AXDTI-inconsistent).

Let $\cM$ be a $T_I\cup \EUF$-model of $A_2\wedge B_2$.
$\cM$ is a two-sorted structure (the sorts are \INDEX and \ELEM)  endowed for every array  constant $d$ occurring in $A\cup B$ of a function $d^\cM:\INDEX^\cM \longrightarrow \ELEM^\cM$. 
In addition, $\INDEX^\cM$ is a model of $T_I$. 
 We list the properties of $\cM$ that comes from the fact that our Steps 1-2 have been  applied;
 below, we denote by $k^\cM$ the element of $\INDEX^\cM$ assigned to an index constant $k$ 
 (thus, if, e.g., $k,l$ are index constants, $\cM\models k=l$ is the same as $k^\cM=l^\cM$):
\begin{compactenum}
 \item[{\rm (a)}] we have that $\cM\models \bigwedge A_1(\cI_A)$ (where $\cI_A$ is the set of \alocal constants) and $\cM\models \bigwedge B_1(\cI_B)$ (where $\cI_A$ is the set of \blocal constants): this is because $(A_1, A_2)$ and $(B_1, B_2)$ are 0-instantiated by Step 2;
\item[{\rm (b)}] 
for \abcommon array variables $c_1, c_2$, we have that $A_1\cap B_1$ contains a literal of the kind $\diff_n(c_1, c_2)= k_n$ for $n\leq N$; suppose that $\cM \models l_{c_1}= l_{c_2}$\footnote{ Recall that  $l_{c_1}, l_{c_2}$ are the \abcommon constants such that the literals 
 $\len{c_1}=l_{c_1}, \len{c_2}=l_{c_2}$ belongs to $A_1\cap B_1$.} and that $k$ is an index constant such that 
 $\cM\models k\neq l$ for all \abcommon index constant $l$; then, we can have $\cM \models c_1(k)\neq c_2(k)$   only when $\cM \models k< k_N$: this is because Step 1 has been  applied and because of (a).
\end{compactenum}

We expand $\cM$ to an \AXEXTTI-structure $\cN$ and endow it with an assignment to our \alocal and \blocal variables, in such a way that all \diff\ operators mentioned in $A_1, B_1$ are defined and all formulae in $A, B$ are true. In view of Theorem~\ref{thm:extension}, this structure can be expanded to the desired full model of \AXDTI.
We take $\INDEX^\cN$ and $\ELEM^\cN$ to be equal to $\INDEX^\cM$ and $\ELEM^\cM$; the $T_I$-reduct of $\cN$ will be equal to the $T_I$-reduct of $\cM$ and we let $x^\cN=x^\cM$ for all index and element constants occurring in $A\cup B$. 
$\ARRAY^\cN$ is the set of all positive support functions from $\INDEX^\cN$ into $\ELEM^\cN$. 
The interpretation of \alocal and \blocal constants of sort \ARRAY is more subtle. We need a d\'etour to 
 introduce it.

Let $k$ be an  index constant s.t.
 $\cM\models k\neq l$ for all \alocal index constant $l$;
we introduce an equivalence relation  
$\equiv_k$ on the set of \alocal array variables as follows: $\equiv_k$ is the smallest equivalence relation that contains all pairs $(a_1,a_2)$ such that $\cM\models l_{a_1}=l_{a_2}$ and moreover an atom of one of the following two 
kinds belongs to $A^0_1$: (I) $a_1= wr(a_2, i, e)$; (II) $\diff(a_1,a_2)= l$, for an $l$ such that $\cM\models l<k$.

\vskip 2mm\noindent
\textbf{Claim}: \emph{if $c_1, c_2$ are \abcommon and $c_1\equiv_k c_2$, then $c_1^\cM(k^\cM)=c_2^\cM(k^\cM)$.}
The claim is proved 
by preventively showing, for every \alocal constants $a_1,a_2$ such that $a_1\equiv_k a_2$,
 that the number of the \alocal constants $j$
such that $\cM\models k<j$ and $a_1^\cM(j^\cM)\neq a_2^\cM(j^\cM)$ is less or equal to $N_A<N$.
In the proof (see the appendix for details),
properties (a) and (b) above and the fact that $N:=1+max(N_A, N_B)$ play a crucial role.

In order to interpret \alocal constants of sort \ARRAY, we assign to an \alocal constant $a$ of sort \ARRAY the function $a^\cN$ defined as follows for every $i\in \INDEX^\cN$:
\begin{compactenum}[(\dag-i)]
\item if $i$ is equal to $k^\cM$, where $k$ is an \alocal constant, then $a^\cN(i):=a^\cM(k^\cM)$;
\item if $i$ is different from $k^\cM$ for every \alocal constant $k$, but 
$i$ is equal to $k^\cM$
for some (necessarily \bstrict) index constant $k$  and 
 there is an \abcommon array variable $c$ such that $c\equiv_k a$,
then $a^\cN(i)$ is equal to $c^\cM(k^\cM)$ (this is well-defined thanks to the claim); 
\item
in the remaining cases,
 $a^\cN(i)$ is equal to $el^\cM$ or $\bot^\cM$ depending whether 
$\cM\models 0\leq i \wedge i\leq l_a$ holds or not.
\end{compactenum}

It can be now shown that 
\begin{compactenum}
 \item[($*$-$A$)] \emph{all $a^\cN$ are positive-support functions and  all \formulae from $A_1\cup A_2$ are true in $\cN$}.
\end{compactenum}
In fact, \formulae in $A_2$ are Boolean combinations of \alocal atoms of the kind~\eqref{phi2}: these are $T_I\cup \EUF$-atoms and, due to their shape, each of them is true in $\cM$ iff it is true in $\cN$. 
 Moreover, formul\ae\ in $A_1$ are 
 all $wr, \diff$ and $ \len{-}$-atoms. 
 They are true in $\cN$ because of the 0-instantiation performed by Step 2 (see (a) above). Full details are in the appendix.

The  assignments to the \blocal array variables $b$ are  defined analogously, so that 
\begin{compactenum}
 \item[($*$-$B$)] \emph{all $b^\cN$ are positive-support functions and  all \formulae from $B_1\cup B_2$ are true in $\cN$}.
\end{compactenum}

There is however one important point to notice.
 For all \abcommon constants $c$ of sort \ARRAY,
 our specification of  $c^\cM$ \emph{does not depend} on the fact that we use the above definition for \alocal or for \blocal array constants: to see this,
 notice that $c\equiv_k c$ holds in case ($\dag$-ii) is applied. This remark concludes the proof. 
\end{proof}

It is not difficult to see that the
 quantifier-free $T_I\cup \EUF$-interpolation problem 
 generated by our algorithm is 
 of polynomial size. Notice however that  $T_I\cup \EUF$-interpolation is at least exponential.

\section{Further Related Work and Conclusions}\label{sec:conclusions}

We introduced two theories of arrays, namely the theory \AXDTI of contiguous arrays with maxdiff and its extension \estesa that also supports `constant' arrays'.
These theories are strictly more expressive than McCarthy's theory and the other variants studied in the literature: notably, strong length of arrays is definable, and inside it arrays are fully defined in every memory location. We proved that \estesa admits general interpolation by showing that its models are strongly amalgamable; the existence of (plain) amalgams also implies that \AXDTI has (plain) interpolants. 
 We also studied the SMT problem for \AXDTI and showed through instantiations techniques that it is decidable. Finally, we provided a general 
 algorithm for computing \AXDTI quantifier-free interpolants that relies on a polynomial reduction to the problem of computing general interpolants for the index theory. Differently from previous algorithms, this procedure avoids full instantiation of terms.
 
 One future research direction regards the implementation of this procedure, which is still missing. In the last decade, some implemented approaches have been introduced to compute interpolants for different theories, by relying on different techniques. For the significant theory of \EUF, in~\cite{cesare} the  DPT prover, which implements a method that exploits colored congruence graphs for extracting interpolants, has been shown to produce simpler and smaller interpolants than other solvers.  For more complex theories, 
 in \cite{mcmillan-z3-interpolation} McMillan proposed
an interpolating proof calculus to compute interpolants via 
refutational proofs obtained from the \textsc{z3} SMT-solver. It is worth mentioning his approach because it takes advantages from the flexibility of the \textsc{z3} solver to deal with several theories and their combination: it makes use of a 
secondary interpolation engine in order to `fill the gaps' 
of refutational proofs introduced by \emph{theory lemmas}, 
which are specific \formulae derived by the satellites theories encoded in \textsc{z3}.
This secondary engine only needs an
interpolation algorithm for \verb|QF_UFLIA|. 
This approach can be used to computes interpolants for array theories, but since they use 
quantified \formulae,  the method can generate quantified
\formulae.

Concerning in particular array theories, another notable approach for computing interpolants is due to the authors of \cite{HS18}, which exploited the proof tree preserving interpolation
scheme from \cite{JuHoe2013} to construct interpolants via a
resolution proof.  This approach is capable of handling, following the nomenclature of our paper, \abcommon literals but not
\abcommon terms: 
 for this reason, \emph{weakly equivalences} between arrays are there introduced in order
to cope with those cases.  This method supports the use of the
$\diff$ operation between arrays in order to compute quantifier-free
interpolants, but its semantic interpretation is undetermined as in~\cite{lmcs_interp}. 


 In~\cite{SMTpaper}, the authors presented AXDInterpolator~\cite{Jose}, an implementation of the interpolation algorithm
from~\cite{ourfossacs}, which allows the user to choose \textsc{z3},
\textsc{Mathsat}, or \textsc{SMTInterpol} (\cite{SMTinterp}) as the underlying interpolation engines. In order to show its feasibility, it was tested against a benchmark based on C programs from the ReachSafety-Arrays and MemSafety-Arrays tracks of SV-COMP~\cite{svcomp}.
Since many C programs from ~\cite{svcomp} require the usage of array length (and, in particular, strong length) we plan to develop a tool that implements the new algorithm for contiguous arrays presented in this paper. 


There is still a question concerning our interpolation algorithm that needs to be investigated:
extending the algorithm to the theory \estesa (with constant arrays in the language). In order to handle constant arrays, 
the construction of Theorem~\ref{thm:al} is still appropriate, except for the fact that condition ($\dag$-ii)
should not be applied to define $\Const(i)^\cN$ when $i$ is an \astrict constant such that $i^\cM$ is equal to
$j^\cM$ for some \abcommon $j$. To avoid this, one could
introduce right after Step 1 some form of guessing for equalities between index constants: however, such a guessing
(based on colorings) would create branches and consequently would not produce a polynomial instance of a $T_I\cup \EUF$-interpolation problem in Step 3. 
This issue needs further analysis.

Finally, although quite challenging, it would be interesting to extend our interpolation results also to array theories combined with cardinality constraints, similar to those introduced, e.g., in~\cite{elena},\cite{kunkak}.

\bibliographystyle{plain}
\bibliography{biblio}

\newpage
\appendix

\onecolumn

\section{Proofs of results from Section~\ref{sec:embeddings}}

\vskip 2mm\noindent
\textbf{Lemma~\ref{lem:dependency}}  \emph{
  Let $\cN$, $\cM$ be models of  \AXEXTTI 
  such that $\cM$ is a
  substructure of $\cN$.  For every $a,b\in\ARRAY^\cM$, we have that
  \begin{eqnarray*}  
    \cM\models a\sim b & \mbox{ {\rm iff} } &
    \cN\models a\sim b.
  \end{eqnarray*}
}
\vskip 1mm
\begin{proof}  
The left-to-right side is trivial because if $\cM\models a\sim b$ then 
$a$ and $b$ have equal length in $\cM$ and in $\cN$ too because length is preserved; moreover, $\cM\models a=wr(b, I, E)$,
where $I\coincide i_1, \ldots, i_n$ is a list of costants (naming elements of $\cM$) of sort
$\INDEX$, $E \coincide e_1, \ldots, e_n$ is a list of costants (naming elements of $\cM$) of sort
$\ELEM$, and $wr(b, I, E)$ abbreviates the term $wr(wr(\cdots wr (b,
i_1, e_1) \cdots), i_n, e_n)$.  Thus,
also $\cN\models a=wr(b, I, E)$ because $\cM$ is a substructure of
$\cN$. Vice versa, suppose that $\cM\not\models a\sim b$. This
means that either $\len{a}\neq \len{b}$ or that there are infinitely many $i\in\INDEX^\cM$ such that
$rd^\cM(a,i)\neq rd^\cM(b,i)$. Since $\cM$ is a substructure of $\cN$, these conditions holds in $\cN$ too.
\end{proof}

\begin{lemma}\label{lem:diffsim} 
Let $\cM$ be a model of $\AXEXTTI$ and let  $a,a',b,b'\in \ARRAY^\cM$; if  
 $a\sim_\cM a'$, $b\sim_\cM b'$ and $\diff_k(a',b')$ is defined for every $k$, then  $\diff(a,b)$ is also defined.
\end{lemma}
\begin{proof}
 Notice first that, from $a\sim a'$ and $b\sim b'$, it follows that $\len{a}=\len{a'}$ and $\len{b}=\len{b'}$. In case $\len{a}\neq\len{b}$, we have that $\diff(a,b)=\max\{\len{a},\len{b}\}$ (see Lemma~\ref{lem:easy}), which implies that $\diff(a,b)$ is defined. 
 Hence, the relevant case is when we have  $l=\len{a}=\len{b}=\len{a'}=\len{b'}$ and $a'\not\sim b'$ (if
 $a'\sim b'$, then we have also $a\sim b$
 and the maximum of the finitely many indexes where 
 $a,b$ differ is $\diff(a,b)$). 
 Then for the  infinitely many indexes $j_k=\diff_k(a',b')$ we have $a'(j_k)\neq b'(j_k)$; for at least one of such $j_k$ we must also have $a(j_k)\neq b(j_k)$ because $a\sim a'$ and $b\sim b'$. 
Consider now the indexes in $[j_k,l]$: 
in this interval, the pair of arrays $a,b$  differs on at least one but at most finitely many indices (because $a,a'$ 
differs on finitely many indices there and so do the pairs $b, b'$ and $a',b'$), so the biggest one such index will be $\diff(a,b)$.
\end{proof}

\begin{lemma}
\label{aggiungoelem}
Let $\cM$ be a model of $\AXEXTTI$. There exist a model  $\cN$ of $\AXEXTTI$ and a \diff-faithful embedding  $\mu:\cM\rightarrow \cN$ such that 
the restriction of $\mu$ to the sort $\ELEM$ is not surjective.
In addition, if $\cM$ is a model of $\extestesa,\AXDTI$ or of \estesa, so it is $\cN$.
\end{lemma}
\begin{proof}
To build $\cN$ it is sufficient to put:
\begin{itemize}
    \item $\INDEX^\cN=\INDEX^\cM$,
    \item  $\ELEM^\cN=\ELEM^\cM\cup \{e\}$ where $e\notin \ELEM^\cM$,
    \item  $\ARRAY^\cN$ 
    consists of the positive-support functions $a:\INDEX^\cN\longrightarrow
    \ELEM^\cN$ for which there exist $a'\in \ARRAY^\cM$ such that $a\sim a'$.
\end{itemize}
Now notice that if \diff\ is totally defined in $\cM$, so it is in $\cN$. In fact, this follows from the definition of $\ARRAY^\cN$  and Lemma~\ref{lem:diffsim}:  
if $a\sim a'$, $b\sim b'$ and $\diff_k(a',b')$ is defined for every $k$, then  $\diff(a,b)$ is also defined by the previous lemma. The claim  is proved  in the same way  for all the above mentioned  array theories  (notice that in case the signature includes the $\Const$ symbol, $\mu$ trivially preserves it).
\end{proof}

\vskip 2mm\noindent
\textbf{Lemma~\ref{immersionediff}}  \emph{Let $\cM$ be a model of $\AXEXTTI$ (resp. of \extestesa) and let
$a,b\in \ARRAY^\cM$ be such that $\diff^{\cM}(a,b)$ is not defined. 
Then there are a model $\cN$ of $\AXEXTTI$ (resp. of \extestesa) and a \diff-faithful embedding
$\mu:\cM\rightarrow \cN$ such that $\diff^\cN(a,b)$ is defined. 
}
\vskip 1mm
\begin{proof} Thanks to Lemma~\ref{aggiungoelem}, 
we can assume that 
$\ELEM^\cM$ 
has at least an element
 $e$  (different from $\bot^\cM, \el^\cM$). Notice that we must have $\len{a}=\len{b}$, otherwise $\diff(a,b)$ is defined and it is $\max(\len{a},\len{b})$ according to Lemma~\ref{lem:easy}.

 Let $I=\{i\in \INDEX^\cM\mid a(i)\neq b(i)\}$ be the set of indices without maximum element (hence infinite) where they differ. 
 Let $\downarrow I := \{j\in \INDEX^\cM \mid \exists i\in I,\ j\le i\}\supseteq I$.
 Notice that  the condition
 $$
 (+)~~ ``\exists i\in I\,\forall j\in I\, ( j\geq i \to x(j)=\el)"
 $$
 cannot be satisfied  both for $x=a$ and $x=b$: indeed, if this were the case, assuming w.l.o.g. that $i_a\leq i_b$ (where $i_a$ and $i_b$ are the witnesses for the existentially quantified index $i$ in $(+)$ for $x=a$ and $x=b$ respectively), we would have that 
 $a(j)=el^\cM=b(j)$ for all $j\geq i_b$, $j\in I$, 
 which is a contradiction with the definition of $I$. 
  In case one of them satisfies it, we assume it is $b$.
  
Let
$\Delta$ be the Robinson diagram of the $T_I$-reduct of $\cM$ and let $k_0$
be  a new constant; let us introduce the set
$$\Delta' := \Delta\cup\{ i<k_0 \mid i \in \downarrow I\} \cup
\{k_0<i \mid i \in \INDEX^\cM\setminus \downarrow I\}.
$$
By the compactness theorem for first order logic and since $I$ is infinite, the set 
 $\Delta'$ turns out to be consistent. In fact, if
  $\Delta'$ were inconsistent, 
  then there would exist a finite subset of it not admitting a model.
  However, a finite subset of
   $\Delta'$ can contain constraints only for a finite number of index constants $d$ occurring in $\Delta$, $i\in  \downarrow I$, $i'\in \INDEX^\cM\setminus \downarrow$ and $k_0$. Such constraints can be verified  inside the $T_I$-reduct of $\cM$ itself: to interpret the additional constant  $k_0$, it is sufficient to 
use the fact that $I$ contains arbitrarily large indexes  and the fact that the definition of
    $\downarrow I$ implies that
$$\forall i\in \downarrow I,\ \forall j\in \INDEX^{\cM}\setminus \downarrow I,\quad i<j. $$

By Robinson Diagram Lemma, there exists a model $\cA$ of $T_I$ extending the 
$T_I$-reduct of  $\cM$; such $\cA$ contains in its support an element  $k_0$ such that
$$
\forall i\in \downarrow I,\ i<k_0,
$$ $$
\forall i\in \INDEX^\cM\setminus \downarrow I,\ k_0<i.
$$

We now take 
 $\ELEM^\cN=\ELEM^\cM$, $\INDEX^\cN=\INDEX^\cA$; we let also  $\ARRAY^\cN$
 to be the set of all positive-support functions from
  $\INDEX^\cN$ into $\ELEM^\cN$ (notice that this $\cN$ is trivially also a model of \extestesa).
  We observe  that $k_0< \len{a}^\cM$ and recall that $\len{a}^\cM=\len{b }^\cM$.
  
Let us now define the embedding 
 $\mu:\cM\rightarrow \cN$; at the level of the sorts \INDEX and \ELEM, we use inclusions. For the \ARRAY sort, we need to specify the value 
  $\mu(c)(k)$ for $c\in \ARRAY^\cM$ and $k\in \INDEX^\cN\setminus \INDEX^\cM$
  (for the other indices we keep the old $\cM$-value to preserve the read operation). Our definition for $\mu$ must preserve the maxdiff index (whenever already defined in $\cM$) and must guarantee that  $\diff^\cN(\mu(a),\mu(b))=k_0$ (by construction, we have $k_0>0$). 
  For a generic array 
   $c\in \ARRAY^\cM$, we operate as follows:
\begin{enumerate}
    \item if $\len{c}^{\cM}<k_0$ we put $\mu(c)(k_0)=\bot^\cM$, otherwise:
    \item if the condition $(\star)$ below holds, we put $\mu(c)(k_0)=e$,
    \item if such condition does not hold, we put $\mu(c)(k_0)=el^\cM$.
\end{enumerate}
The condition  $(\star)$ is specified as follows:
\begin{description}
 \item[$(\star)$] ~~
there is $i\in I$ such that for all $j\in I, j\geq i$ we have $c(j)=a(j)$.
\end{description}
\noindent
For all the remaining indexes $k\in \INDEX^\cN\setminus (\INDEX^\cM\cup \{k_0\})$ we put
$$
    \mu(c)(k)=\begin{cases} 
\bot^\cM, & \mbox{if $k \notin [0, \len{c}^\cM]$}\\
\el^\cM, & \mbox{if $k \in [0, \len{c}^\cM]$}
    \end{cases}
    \,\,\,\,\,\,\,   \quad \quad     \eqref{indicinonk0} $$
Notice that   we have $\mu(a)(k_0)=e\neq \el^\cM=\mu(b)(k_0)$ (the last equality holds because $I$ is infinite and does not have maximum, hence condition  $(\star)$ holds for $a$ but not for $b$). 
   In addition:
\begin{itemize}
    \item for all $i\in\INDEX^\cM$ such that $k_0<i$, 
    we have $i\notin\, \downarrow I$, according to the construction of $k_0$ and
    consequently $i\notin I$, that is $a(i)=b(i)$;
    \item for all $i\in\INDEX^\cN\setminus (\INDEX^\cM\cup\{k_0\})$ such that $k_0<i$, since we have  $\len{a}^\cM=\len{b}^\cM$, we get
    $$\mu(a)(i)=\bot^\cM~{\rm iff}~ \mu(b)(i)=\bot^\cM,$$
    $$\mu(a)(i)=\el^\cM ~{\rm iff}~ \mu(b)(i)=\el^\cM.$$
\end{itemize}
Hence, we can conclude that  $\diff^\cN(\mu(a),\mu(b))$ is defined and equal to $k_0$.

We only need to check that our $\mu$ preserves $rd, \len{-}, wr$, constant arrays and \diff\ 
(whenever defined).

\vspace{5pt}
 The operation $rd$ is preserved because $\mu$ acts as an inclusion for indexes and elements and because we have 
 $\mu(c)(k)=c(k)$ if $k\in \INDEX^\cM$.
 
\vspace{5pt} Concerning length, we have
 $\len{\mu(c)}^\cN=\len{c}^\cM$ because of \eqref{indicinonk0} and because of the above definition of $\mu(c)(k_0)$ (recall that $k_0>0$). 

\vspace{5pt}
Concerning write operation, we prove that for all $c\in \ARRAY^{\cM}$,
$i\in \INDEX^{\cM}\cap [0,\len{c}]$ and $e'\in \ELEM^{\cM}\setminus\{\bot^\cM \}$  
we have
$$
\mu(wr(c,i,e'))=wr(\mu(c),i,e').
$$
Remember that we have $\len{wr(c,i,e')}=\len{c}=\len{\mu(c)}$.
\begin{itemize}
    \item For $k\neq i$ in $\INDEX^\cM$
    $$
    \mu(wr(c,i,e'))(k)=wr(c,i,e')(k)=c(k)
    $$ $$
    wr(\mu(c),i,e')(k)=\mu(c)(k)=c(k);
    $$
    \item For $k=i$
    $$
    \mu(wr(c,i,e'))(i)=wr(c,i,e')(i)=e'
    $$ $$
    wr(\mu(c),i,e')(i)=e';
    $$
    \item For $k=k_0>\len{c}$ the claim  follows immediately from the definition;
    \item For $k=k_0<\len{c}$, $(\star)$ holds for $c$ iff it 
    holds for $wr(c,i,e')$, because $I$ is infinite.
     Hence we have
    $$
    \mu(wr(c,i,e'))(k_0)= e~~{\rm iff}~~  \mu(c)(k_0)=e.
$$ 
    \item For $k\in \INDEX^\cN \setminus (\INDEX^\cM \cup \{k_0\})$, the claim is clear from the definition and from the fact that  $wr$ preserves length. 
\end{itemize}

\vspace{5pt}
For constant arrays, we must show only that $\mu(\Const(i))(k_0)=\el^\cM$ in case $k_0<i$: this is clear, because $a$ does not satisfy (+), hence  $(\star)$ does not hold for $\Const(i)$. 

\vspace{5pt}
Let us now finally consider the \diff\ operation and let us prove that if 
$\diff^\cM(c_1,c_2)$ is defined, then  $\diff^N(\mu(c_1),\mu(c_2))$ is also defined and equal to it. 
Assume that $\diff^\cM(c_1,c_2)$ is defined; since $\mu$ preserve length, the only relevant case,
in view of Lemma~\ref{lem:easy}, is when we have $\len{c_1}=\len{c_2}$; since the values of $c_1, c_2$ on indexes from $\cM$ are preserved, taking in mind~\eqref{indicinonk0} (in particular, that for $k \in \INDEX^\cN \setminus (\INDEX^\cM \cup \{k_0\})$ we have $\mu(c_1)(k)=\mu(c_2)(k)$), 
we only have to exclude that we have 
$$
\diff^{\cM}(c_1,c_2)<k_0 ~{\rm and}~\mu(c_1)(k_0)\neq \mu(c_2)(k_0)
$$
If this is the case, we have, e.g., 
 $\mu(c_1)(k_0)=e\neq \el^\cM=\mu(c_2)(k_0)$ ($k_0<\len{c_1}=\len{c_2}$), which implies that $(\star)$ holds for $c_1$ but not for $c_2$. 
 However,  
 it cannot be that $(\star)$  holds for only one among $c_1,c_2$. The reason for this is as 
 follows.
 Indeed, if $i\in I$ is the index that witnesses $(\star)$ for $c_1$, then $c_1(j)=a(j)$ for all indexes $j\in I$ such that $j\geq i$ (which are infinitely many). Since $I$ is infinite and without maximum and since $\diff^{\cM}(c_1,c_2)<k_0$, 
 we must have $\diff^{\cM}(c_1,c_2)\in\, \downarrow I$ by the definition of $\Delta$, so 
 there must be infinitely many indices in $I$ bigger than 
 $\diff^{\cM}(c_1,c_2)$ and arbitrarily large, which means in particular that there exists an index $i'\in I$ such that $i'\geq i$ and $i'> \diff(c_1,c_2)$. This $i'$ witnesses $(\star)$ for $c_2$, as wanted. This concludes the proof. 
 \eop
\end{proof}

\section{Proofs of results from Section~\ref{sec:strongamalg}}

We prove here that the model $\cM$ introduced in Section~\ref{sec:strongamalg} is a \estesa-amalgam (and also a \AXDTI-amalgam) for the models $\cM_1$ and $\cM_2$ with the common substructure $\cA$.
\subsection*{Requirements check for the amalgamated model}
\begin{proof}
We need to prove that the functions 
 $\mu_i$: (i) are well-defined, (ii) are injective, (iii) preserve $\len{-}$,
 (iv) preserve 
 $rd$ and $wr$, (v) preserve  $\diff$, (vi) satisfy the condition $ \mu_{1} \circ f_1=  \mu_{2} \circ f_2 $, (vii) preserve constant arrays (for the statement about \estesa).

\begin{enumerate}
    \item [(i)] Since $\INDEX^\cM$ is a strong amalgam of $\INDEX^{\cM_1}$
    and $\INDEX^{\cM_2}$, the case distinctions we made for defining $\mu_i(a)(k)$ are non-overlapping and exhaustive.

\noindent We now show that if, for $i=1,2$ and
 $k\in \INDEX^{\cM_{3-i}}\setminus \INDEX^\cA$ and $a\in \ARRAY^{\cM_i}$,
 the relation ($2\star$) holds relatively to two different pairs of arrays  $(c_1,b_1)$, $(c_2,b_2)$ from $\ARRAY^\cA \times \ARRAY^{\cM_i}$, 
 then we nevertheless have 
 $f_{3-i}(c_1)(k)=f_{3-i}(c_2)(k)$ 
 (this proves the consistency of the definition).
 For symmetry, let us consider only the case $i=1$. 
 Since the index ordering is total, let us suppose that we have for instance
  \begin{equation}\label{eq:cp}
   k>\diff^{\cM_1}(b_1,f_1(c_1)) \ge \diff^{\cM_1}(b_2,f_1(c_2)).
  \end{equation}
By the transitivity of
 $\sim^{M_1}$, 
 %
  $b_1$ and $b_2$ differ on finitely many indices,
 hence we can consider the finite sets 
$$ J:= \{j\in \INDEX^\cA \mid \ b_1(j)\neq b_2(j), j>\diff^{\cM_1}(b_1,f_1(c_1))\}$$
$$ E:= \{b_1(j) \mid\ j\in J\}\subseteq \ELEM^\cA.$$

Let now pick $c:=wr(c_2,J,E)$; then  $c\sim^{\cA}c_2$. Since $f_2$ is an embedding, 
 we have 
  $f_2(c)(k)=f_2(c_2)(k)$ for all $k\in \INDEX^{\cM_2}\setminus \INDEX^\cA$.
  Suppose we have also
\begin{equation}
\label{supp1}
    \diff^\cA(c_1,c)\le \diff^{\cM_1}(f_1(c_1),b_1).
\end{equation}
Then
$$ \diff^\cA(c_1,c)\le \diff^{\cM_1}(f_1(c_1),b_1)<k $$
and consequently also the desired equality 
$$ f_2(c_1)(k)=f_2(c)(k)=f_2(c_2)(k) $$
follows.

In order to prove
 (\ref{supp1}), we consider  $j\in \INDEX^\cA$ with $j>\diff^{\cM_1}(f_1(c_1),b_1)$ 
 and show that we have 
  $c(j)=c_1(j)$. 
  Suppose not, i.e. that 
  $c(j)\neq c_1(j)$; then  we cannot have  $j\in J$ otherwise, by definition of $c$ and $E$, we would have 
  $c(j)=b_1(j)=f_1(c_1)(j)=c_1(j)$ ($f_1$ is inclusion for indexes 
  and $J\subseteq\INDEX^\cA$), contradiction. Hence, we have $j\notin J$, 
 so $c(j)=c_2(j)$ and $b_1(j)=b_2(j)$. Now remember that  
$$ j>\diff^{\cM_1}(b_1,f_1(c_1))\ge \diff^{\cM_1}(b_2,f_1(c_2)) ,$$
hence
$$ c_1(j)=b_1(j),\ c_2(j)=b_2(j)$$
$$ c(j)=c_2(j)=b_2(j)=b_1(j)=c_1(j)$$
thus getting an absurdity.
    
\item[(ii)]  Injectivity of $\mu_1$ and $\mu_2$ is immediate.
    
\item [(iii)] In order to prove that $\len{-}$ is preserved, it is sufficient to show that for every $a\in \ARRAY^{\cM_1}$ and for all $ k\in \INDEX^\cM$, we have 
$$
\mu_1(a)(k)\neq \bot \leftrightarrow 0\le k \le \len{a}^{\cM_1}.
$$
The only relevant case is when $k\in \INDEX^{\cM_2}\setminus \INDEX^\cA$ and ($2\star$) holds. In such a case, we have two possibilities:
    \begin{compactenum}
        \item [$\len{b}^{\cM_1}=\len{c}^{\cA}$:] in this case, since $b\sim^{\cM_1}a$, we have  $\len{a}^{\cM_1}=\len{b}^{\cM_1}=\len{c}^{\cA}=\len{f_2(c)}^{\cM_2}$ ($f_2$ is an embedding), thus getting what we need for $k$;
        \item[$\len{b}^{\cM_1}\neq\len{c}^{\cA}$:] in this case $k>\diff^{\cM_1}(b,f_1(c))=\max\{\len{b}^{\cM_1},\len{c}^\cA\}$ by Lemma~\ref{lem:easy}.
        Since $a\sim b$ implies $\len{a}=\len{b}$, the definition of $\mu_1$ produces $\mu_1(a)(k)=f_2(c)(k)=\bot$ (the last identity holds because $k>\len{c}^{\cA}$),
         which is as desired because
         $k>\len{b}^{\cM_1}=\len{a}^{\cM_1}$ too.
    \end{compactenum}
    
\item[(iv)] The fact that 
 $rd$ and $wr$ operations are preserved is easy (notice in particular that, if   ($2\star$) holds for $a$ via the pair $(c,b)$, 
 then the same pair guarantees 
  ($2\star$) for arrays of the kind $wr(a,i,e)$).

\item[(v)] Again we limit to the case of $\mu_1$ for symmetry.
We need to show that for every $a_1,a_2\in \ARRAY^{\cM_1}$, we have $\mu_1(\diff^{\cM_1}(a_1,a_2))=\diff^{\cM}(\mu_1(a_1),\mu_1(a_2))$. Notice that, if we call $j$ the index $\diff(a_1,a_2)\in \INDEX^{\cM_1}$, by definition of $\mu_1$ on array applied to indexes in $\INDEX^{\cM_1}$, we have that  $\mu_1(a_1)(j)= a_1(j)\neq a_2(j)=\mu_1(a_2)(j)$. Hence, in order to conclude, it is sufficient to show that, given  $k\in \INDEX^\cM$  such that $k>\diff^{\cM_1}(a_1,a_2)$, the equality $\mu_1(a_1)(k)=\mu_1(a_2)(k)$ holds. 
Notice first that we can always reduce to one of the following three cases
\begin{enumerate}
    \item [(a)] $\len{a_1}<\len{a_2}=\diff^{\cM_1}(a_1,a_2)$;
    \item [(b)] $\diff^{\cM_1}(a_1,a_2)<\len{a_1}=\len{a_2}$;
    \item [(c)] $\len{a_1}=\len{a_2}=\diff^{\cM_1}(a_1,a_2)$.
\end{enumerate}
We now show that 
$\mu_1(a_1)(k)=\mu_1(a_2)(k)$.
\begin{itemize}
    \item If $k \in \INDEX^{\cM_1}$:
    $$ \mu_1(a_1)(k)=a_1(k)=a_2(k)=\mu_1(a_2)(k).$$
    \item If $k\notin \INDEX^{\cM_1}\cup \INDEX^{\cM_2}$, we analyze the three cases separately:
    \begin{enumerate}
        \item [Case(a)]: then $k\notin [0,\len{a_i}]$ for $i=1,2$. We have $\mu_1(a_1)(k)=\bot=\mu_1(a_2)(k)$;
        \item[Case (b)]: we have 
        $$ \mu_1(a_1)(k)=\bot \Leftrightarrow \mu_1(a_2)(k)=\bot$$
        $$ \mu_1(a_1)(k)=el \Leftrightarrow \mu_1(a_2)(k)=el; $$
        \item [Case (c)]: similarly to  (a), we have  $k\notin [0,\len{a_i}]$ for $i=1,2$.
    \end{enumerate}
    \item If $k\in \INDEX^{\cM_2}\setminus \INDEX^{\cA}$ and ($2\star$) does not hold neither for $a_1$ nor for $a_2$, the argument is the same as in the previous case.
    
    Otherwise, suppose that ($2\star$) holds for, say,  $a_1$ 
    as witnessed by the pair $(c_1,b_1)$. Then we get 
     $\mu_1(a_1)(k)=f_2(c_1)(k)$. We prove that 
     ($2\star$) holds for  $a_2$ too and that we have $\mu_1(a_1)(k)=f_2(c_1)(k)=\mu_1(a_2)(k)$.  
    
    Since $b_1$ and $a_1$ differ on finitely many indices inside $\cM_1$, we can consider the finite sets
    $$ I:=\{i\in \INDEX^{\cM_1}\mid \ b_1(i)\neq a_1(i),\ i>\diff^{\cM_1}(a_1,a_2)\}, $$
    $$
    E:=\{b_1(i)\mid\ i\in I\}\subseteq \ELEM^{\cM_1}
    $$
    and the array $\hat{b}:=wr(a_2,I,E)$; for this array, we obviously have   $a_2\sim^{\cM_1} \hat{b}$. If we also have 
    \begin{equation}
    \label{supp2}
        \diff^{\cM_1}(b_1,\hat{b})\le \diff^{\cM_1}(a_1,a_2)
    \end{equation}
    then we get:
    $k>\diff^{\cM_1}(a_1,a_2)$ e $k>\diff^{\cM_1}(b_1,f_1(c_1))$ (the latter is from ($2\star$)). Hence:
    $$ k>\max\{\diff^{\cM_1}(a_1,a_2),\diff^{\cM_1}(b_1,f_1(c_1))\} \ge$$
    $$
    \ge \max\{\diff^{\cM_1}(b_1,\hat{b}),\diff^{\cM_1}(b_1,f_1(c_1))\}\ge$$
    $$
    \ge \diff(\hat{b},f_1(c_1))
    $$
    (the last disequality holds because of the `triangular disequality'~\eqref{eq:triangular} of Lemma~\ref{lem:easy}).
    Hence we obtain ($2\star$) for $a_2$ via the pair given by $c:=c_1$ and $b:=\hat{b}$ (because we have $a_2\sim^{\cM_1}\hat{b}$ and $k>\diff(\hat{b},f_1(c_1))$) and consequently $\mu_1(a_2)(k)=f_2(c_1)(k)$.
    
    It remains to prove 
     (\ref{supp2}); to this aim, let us pick $j\in\INDEX^{\cM_1}$ such that $j>\diff(a_1,a_2)$ and let us show that   $b_1(j)=\hat{b}(j)$. If this is not the case, i.e., if $b_1(j)\neq \hat{b}(j)$, 
     then according to the definition of 
      $\hat{b}$ we have $\hat{b}(j)=a_2(j)$ and $j\notin I$. Hence $\hat{b}(j)=a_2(j)=a_1(j)=b_1(j)$  (the last identity holds because $j\notin I$ and $j>\diff(a_1,a_2)$), absurd. 

\end{itemize}

    \item[(vi)]
    In order to prove $ \mu_{1} \circ f_1  =  \mu_{2} \circ f_2 $, let us consider $c\in \ARRAY^\cA$; let us put  $a_i=f_i(c)$ for  $i=1,2$ and let us check that  $$\mu_1(a_1)(k)=\mu_2(a_2)(k)$$
    holds for all $k\in \INDEX^\cM.$
\begin{itemize}
    \item Case $k\in \INDEX^\cA$: we have $$\mu_1(a_1)(k)=a_1(k)=f_1(c)(k)=c(k)$$ $$\mu_2(a_2)(k)=a_2(k)=f_2(c)(k)=c(k).$$
    \item Case $k\in \INDEX^{\cM_1}\setminus \INDEX^\cA $:
    clearly  ($2\star$) holds for  $a_2$ with $c:=c$ and $b:=a_2$, consequently
    $$ \mu_1(a_1)(k) = a_1(k) =  f_1(c)(k)$$
    $$ \mu_2(a_2)(k) = f_1(c)(k).$$
    \item Case $k\in \INDEX^{\cM_2}\setminus \INDEX^\cA $:
    clearly  ($2\star$) holds for  $a_1$ with $c:=c$ and $b:=a_1$, consequently
    $$ \mu_1(a_1)(k) =  f_2(c)(k)$$
    $$ \mu_2(a_2)(k) = a_2(k) = f_2(c)(k).$$
    \item Case $k\notin (\INDEX^{\cM_1}\cup \INDEX^{\cM_2})$ and $ k\in [0,\len{c}] $:
    we have $$ \mu_1(a_1)(k) = el = \mu_2(a_2)(k). $$
    \item se $k\notin (\INDEX^{\cM_1}\cup \INDEX^{\cM_2})$ and $k\notin [0,\len{c}] $: we have 
    $$ \mu_1(a_1)(k) = \bot = \mu_2(a_2)(k) .$$
\end{itemize}
This completes our case analysis.
\item[(vii)] Here we assume that $\cM_1, \cM_2$ are models of \estesa; we need to show, e.g., that $\mu_1(\Const^{\cM_1}(i))(k)=\el$ for every $k$
such that $0\leq k\leq i$.
Now, if $i\in \INDEX^\cA$ this is obvious, because $\Const^{\cM_1}(i)=f_1(\Const^{\cA}(i))$.  Hence suppose that $i\in\INDEX^{\cM_1}\setminus \INDEX^\cA$; the only possibly problematic case is when $k\in \INDEX^{\cM_2}\setminus \INDEX^\cA$ and ($2\star$) applies, as witnessed by a pair $(b,c)$ 
for $\Const^{\cM_1}(i)$. But we have $\len{b}=\len{\Const^{\cM_1}(i)}=i$ and $\len{f_1(c)}\neq i$ 
(because $\len{f_1(c)}=\len{c}\in \INDEX^\cA$). Then, according to ($2\star$) and recalling Lemma~\ref{lem:easy}, we have
$k>\diff^{\cM_i}(b,f_i(c))=\max\{\len{b}, \len{f_1(c)}\}=\max\{i, \len{f_1(c)}\}\geq i$, contradicting the choice of $k$.
%
\end{enumerate}
\end{proof}


To prove \emph{strong} amalgamation for \estesa we need a couple of lemmas.

\begin{lemma}\label{lem:indexinfinite}
Every model $\cM$ of \estesa can be embedded into a model $\cN$ such that $\INDEX^\cN$ is infinite. 
\end{lemma}

\begin{proof} This is basically due to the fact that $T_I$ is stably infinite.
So let us first embed the $T_I$-reduct of $\cM$ into an infinite model $\cA$ of $T_I$. We define $\cN$ as follows. We let $\ELEM^\cN$ be equal to $\ELEM^\cM$ and the $T_I$-reduct of $\INDEX^\cN$ be equal to $\cA$. We let $\ARRAY^\cN$ be the set of positive support functions from $\INDEX^\cN$ to $\ELEM^\cN$ (the model so built will then be embedded into a full model of \estesa using Theorem~\ref{thm:extension}). We only need to define the embedding $\mu: \cM\longrightarrow \cN$.
This emebdding will be the identity for \INDEX and \ELEM sorts; for arrays, we let $\mu(a)(k)$ be equal to $a(k)$ for $k\in \INDEX^\cM$ and for $k\not\in \INDEX^\cM$, we put $\mu(a)(k)$ equal to  $el^\cM$ or $\bot^\cM$ depending whether we have $k\in[0,\len{a}]$ or not. The proof that $\mu$ preserves all operations is easy.
\end{proof}

Let us call an element $i\in \INDEX^\cM$ of a model $\cM$ of $\estesa$ \emph{finite} iff the set $\{ j\in \INDEX^\cM \mid 0\leq j\leq i\}$ is finite. $Fin(\cM)$ denotes the set of finite elements of $\cM$.

\begin{lemma}
\label{lemmaka} Let $\cA,\cM$ be  models of \estesa and let $f:\cA\longrightarrow \cM$
be an embedding.
Then there exist a third model $\cN$ of  $\estesa$ and an embedding  
 $\nu:\cM\rightarrow \cN$, such that for every $a\in \ARRAY^\cN$ 
 with $\len{a}^\cN\in \nu(f(\INDEX^\cA))$
 one of the following conditions hold:
\begin{enumerate}
    \item there exists $c\in \ARRAY^\cA$ such that $\nu(f(c))\sim^{\cN}a$; 
    \item there exists $k_a\in \INDEX^\cN\setminus \nu(\INDEX^{\cM})$ with $a(k_a)\not \in \nu(f(\ELEM^\cA))$ such that 
    for every $c\in \ARRAY^\cA$ we have  
    $ a(k_a)\neq\nu(f(c))(k_a).$
   %
\end{enumerate}
\end{lemma}

\begin{proof} Because of Lemma~\ref{lem:indexinfinite}, we can freely assume that $\INDEX^\cM$ is infinite.
 If for all 
 $a\in \ARRAY^{\cM}$ whose length comes from $\cA$, there is $c\in \ARRAY^\cA$ such that 
 $f(c)\sim^{\cM} a$ then it is sufficient to take  
 $\cN=\cM$ and the identity as $\nu$. 
 Otherwise, one take a well ordering of the arrays, apply the construction below by tranfinite induction and repeat it $\omega$-times. The union of the chain so built will have the required properties.

\vspace{5pt}
Let $a\in\ARRAY^{\cM}$
be such that $\len{a}$ is from $\cA$ (i.e. such that $\len{a}=f(i)$ for  some $i\in \INDEX^\cA$) and such that there does not exist
$c\in \ARRAY^\cA$ such that 
$f(c)\sim^{\cM} a$. Then $\len{a}$ is not finite, because otherwise we would have that 
$a\sim^{\cM} f(\Const(i))$.\footnote{
Notice that this is the only argument in the whole strong amalgamation proof  requiring the fact that we have $\Const$ in the language.
} Consider the 
diagram $\Delta$ of the $T_I$-reduct of $\cM$ and let $k_a$ a fresh constant; the set 
\begin{equation}
\begin{aligned}
    \Delta' := \Delta \cup \{i< k_a  \mid i\in Fin(\cM)\}
    \cup \\ \cup \{i> k_a  \mid i\in \INDEX^\cM\setminus Fin(\cM)\}.
\end{aligned}
\end{equation} 
is consistent.
Suppose that $\Delta'$ is inconsistent. By compactness, we would have that there exists a finite subset $\Delta'_0$ of $\Delta'$ which is inconsistent too.  $\Delta'_0$ would involve finitely many finite indexes $i_1<\cdots < i_n$ and finitely many infinite indexes $j_1<\cdots<j_m$ and a finite subset $\Delta_0$ of $\Delta$.

If $m=0$, since $\INDEX^\cM$ is infinite, 
%
%
there exist an element $i'\in \INDEX^{\cM} $ large enough, 
so as to get $i_1<\cdots< i_n<i'$ in ${\cM}$. 
If $m>0$, then already in $\INDEX^{\cM}$ there exists an element $i'$ 
such that $i_1<\cdots < i_n<i'$ and such that $i'<j_1<\cdots< j_m$ hold in $\cM$, otherwise $j_1$ would be a finite index. This element $i'$ can interpret the constant $k_a$. In both cases, we conclude that $\Delta'_0$ would be consistent, which is a contradiction.

\noindent By Robinson Diagram Lemma, 
$\Delta$ has a model $\cB$ extending the $T_I$-reduct of $\cM$. 

We let now $\ELEM^\cN=\ELEM^\cM$, $\INDEX^\cN=\INDEX^\cB$ and we let $\ARRAY^N$ 
to be the set of positive-support functions from $\INDEX^\cN$ into $\ELEM^\cN$
(then, in view  of Theorem~\ref{thm:extension}, $\cN$ can be embedded into a full model of \estesa). 
This model contains an element $k_a$ such that for all $i\in \INDEX^\cM$ we have that $k_a\neq i$ and $i< k_a$ iff $i\in Fin(\cM)$. In particular, $k_a<\len{a}$ (because $\len{a}$ is infinite). 

Thanks to Lemma
 \ref{aggiungoelem}, we can freely suppose that $\ELEM^\cN=\ELEM^\cM$ has an element $e$ 
 not belonging to $f(\ELEM^\cA)$ (in particular $e\neq f(\el^\cA)=\el^\cN=\el^\cM $).
 We build $\mu$ as required by the statement of the lemma (more precisely, the $\nu$ required by the lemma will be a chain unions of the $\mu$'s built at each transfinite step as shown below). The \INDEX- and \ELEM-components of $\mu$ will be inclusions. 
 We define $\mu(b)(k)$ for all $b\in\ARRAY^\cM$. If $k\in \INDEX^\cM$, we obviously put  $\mu(b)(k)=b(k)$; in the other cases, the definition is as follows:
\begin{enumerate}
    \item if $b\not\sim^\cM a$, then $\mu(b)(k)=\el^\cN$ or $\mu(b)(k)=\bot^\cN$, depending on whether $k\in [0,\len{b}]$ or not;
    \item if $b\sim^\cM a$ and $k\neq k_a$, then again $\mu(b)(k)=\el^\cN$ or $\mu(b)(k)=\bot^\cM$, depending on whether $k\in [0,\len{b}]$ or not;
    \item if $b\sim^\cM a$ and $k= k_a$,  then $\mu(b)(k)=e$.
\end{enumerate}

We need to show that $\mu$ preserves $rd,wr, \len{-}$, constant arrays and $\diff$.
Preservation of $rd,wr, \len{-}$ are easy; constant arrays are preserved, because we cannot have $\Const^\cM(i)\sim^\cM a$, otherwise $\len{a}=i$, which cannot be because $\len{a}$ is an element from $\INDEX^\cA$ by hypothesis, so that we would have
$f(\Const^\cA(i))=\Const^\cM(i)\sim^\cM a$, contradiction. For preservation of \diff, the problematic case would be the case 
 in which we have $k_a>\diff(b_1,b_2)$, $b_1\sim^\cM a$ and $b_2\not\sim^\cM a$. However, this is impossible because $\diff(b_1,b_2)\in \INDEX^\cM$ and the fact that we have $k_a>\diff(b_1,b_2)$ implies that $\diff(b_1,b_2)$ is finite, which would entail either $b_1\sim^\cM b_2$ or 
$\len{b_1}\neq \len{b_2}$: in the former case, we would have $b_2\sim^\cM a$ and in the latter $k_a >\diff(b_1,b_2)=\max(\len{b_1}, \len{b_2})= \max(\len{a}, \len{b_2})\geq
\len{a}$.

We finally notice that $k_a$ satisfies the requirements of the lemma.
First, $k_a\not\in\mu(\INDEX^\cM)$ and $\mu(a)(k_a)=e\not \in 
\mu(f(\ELEM^\cA))$
hold by construction. Moreover, since for every $c\in \ARRAY^\cA$  we have $ \mu(f(c))(k_a)=\el^\cN$ or $ \mu(f(c))(k_a)=\bot^\cN$ (depending whether $k_a\in [0,\len{c}]$ holds or not), in any case we see that $ \mu(f(c))(k_a)\neq \mu(a)(k_a)$.
\end{proof}

\vskip 2mm\noindent
\textbf{Theorem~\ref{thm:axd_strong_amalg}}  \emph{
 $\estesa$ enjoys the strong amalgamation property.
}
\vskip 1mm
\begin{proof}
We keep the same notation and construction as in the proof of Theorem~\ref{thm:axd_amalg}. However, thanks to 
 Lemma~\ref{lemmaka} we can now suppose (for $i=1,2$) that all arrays 
 $a\in \ARRAY^{\cM_i}$ whose length belongs to $\INDEX^\cA$ are such that one of the following two conditions are satisfied:
\begin{enumerate}
    \item there exists $c\in \ARRAY^\cA$ with $f_i(c)\sim^{\cM_i}a$;
    \item there exists $k_a\in \INDEX^{\cM_i}\setminus \INDEX^\cA$ 
    such that $a(k_a)$ is an element from $\ELEM^{\cM_i}\setminus \ELEM^\cA$ different from all the $f_i(c)(k_a)$, varying
     $c\in \ARRAY^\cA$. 
\end{enumerate}

Let 
$a_i\in \ARRAY^{\cM_i}$ ($i=1,2$) be such that $\forall k\in \INDEX^\cM$ we have
$$
  \mu_1(a_1)(k)=\mu_2(a_2)(k) 
  \quad \quad \quad\eqref{indiciamalgama3}
  $$
Notice that, since $T_I$ has the strong amalgamation property and the $\mu_i$ preserve length, this can only happen if $\len{a_1}=\len{a_2}$ belongs to $\INDEX^\cA$. We look for some
$c\in \ARRAY^{\cA}$ such that  $a_1=f_1(c)$; since $\mu_2$ is injective this would entail  $a_2=f_2(c)$ because
$$
\mu_2(a_2)=\mu_1(a_1)=\mu_1(f_1(c))=\mu_2(f_2(c)),
$$
implying that $\hat{\cM}$ is a strong amalgam, as requested. 

We separate two cases: (i) one of the arrays $a_1, a_2$ satisfy the above condition 2; (ii) both arrays $a_1, a_2$ satisfy the above condition 1.

\begin{enumerate}
    \item [(i)]
    We show that this case is impossible.
     Suppose, e.g., that  $a_1$ satisfies condition 2 in $\cM_1$. 
     Then there exists an index 
      $k_{a_1}$ in $\INDEX^{\cM_1}\setminus \INDEX^\cA$ such that 
       $a(k_{a_1})$ is an element from $\ELEM^{\cM_1}\setminus \ELEM^\cA$
       which is
       different from all the $f_1(c)(k_{a_1})$, varying  $c\in \ARRAY^\cA$. Since we must have $\mu_1(a_1)(k_{a_1})=\mu_2(a_2)(k_{a_1})$ and $\mu_1(a_1)(k_{a_1})$ does not belong to $\ELEM^{\cM_2}$
      (recall that $\ELEM^{\cM_1}\cap \ELEM^{\cM_2}=\ELEM^{\cA}$),
     the value of $a_2$ for the index $k_{a_1}$ ($0<k_{a_1}<|a_1|=|a_2|$) is built according to the rule ($2\star$), because otherwise $\mu_2(a_2)(k_{a_1})$ would be equal to some element in $\ELEM^{\cM_2}$. 
    Let  $(c,b)$ the pair such that
    $$
    c\in \ARRAY^\cA, b\in \ARRAY^{\cM_2},\ b\sim^{\cM_2}a_2,\ k_{a_1}>\diff^{\cM_2}(b,f_2(c))
    $$
    $$
    \mu_2(a_2)(k_{a_1})=f_1(c)(k_{a_1}).
    $$
    Then we have
    $$
    \mu_1(a_1)(k_{a_1})=a_1(k_{a_1}) \neq f_1(c)(k_{a_1})=\mu_2(a_2)(k_{a_1})
    $$
    contradiction.
    
    \item [(ii)] Hence we can have 
    $\mu_1(a_1)=\mu_2(a_2)$ only when both $a_1,a_2$ satisfy condition 1 above. Let us call
    $c_i\in \ARRAY^\cA$ ($i=1,2$) the arrays such that $f_i(c_i)\sim^{\cM_i}a_i$. 
    Then,   the pair $(c_1,f_1(c_1))$  witnesses ($2\star$) for $a_1$ and  for every positive \footnote{
    Notice that if  $k\in\INDEX^{\cM_1}\setminus \INDEX^\cA$ is positive, then 
    $k>\diff(f_1(c),f_1(c))=0$.
    } index  $k\in \INDEX^{\cM_1}\setminus \INDEX^\cA$ (and similarly for $a_2$). We look for 
     $c\in \ARRAY^\cA$ such that $f_1(c)=a_1$. 
     Let us consider the following relations coming from
      (\ref{indiciamalgama3}) and from the definition of $\mu_i$:  
\begin{equation}
\label{diagramma3}
\begin{aligned}
    \forall k\in \INDEX^\cA,\ a_1(k)=a_2(k)\\
    \forall k\in \INDEX^{\cM_1}\setminus \INDEX^\cA,\ a_1(k)=f_1(c_2)(k) \\
    \forall k\in \INDEX^{\cM_2}\setminus \INDEX^\cA,\ f_2(c_1)(k)=a_2(k).\\
    \end{aligned}
\end{equation}
where we used that  $\mu_2(a_2)(k)=f_1(c_2)(k)$ if $k\in \INDEX^{\cM_1}\setminus \INDEX^\cA$, and $\mu_1(a_1)(k)=f_2(c_1)(k)$ if $k\in \INDEX^{\cM_2}\setminus \INDEX^\cA$.

Let us now consider the sets (they are finite because
     $f_2(c_2)\sim^{\cM_2}a_2$)
    $$
    J=\{j\in \INDEX^\cA\mid\ c_2(j)\neq a_2(j)\}\subseteq \INDEX^\cA
    $$
    $$
    E=\{a_2(j) \mid\  j\in J\}\subseteq \ELEM^\cA,
    $$
    and let us put $c=wr(c_2,J,E)$; we check that $c$ is such that $f_1(c)=a_1$.
    \begin{itemize}
    \item If $k\in \INDEX^\cA$: $$f_1(c)(k)=c(k)=wr(c_2,J,E)(k)=a_2(k)=a_1(k)$$
    because  $f_1$ preserves $rd$, by the definition of $J$ and because of the equalities (\ref{diagramma3});
    \item If $k\in \INDEX^{\cM_1}\setminus \INDEX^\cA$: $$f_1(c)(k)=f_1(wr(c_2,J,E))(k)=wr(f_1(c_2),J,E)(k)=f_1(c_2)(k)=a_1(k)$$
    by the definition of $c$, the fact that  $f_1$ is an embedding, because $J \subseteq \INDEX^\cA$ (hence $k\notin J$)
    and because of the equalities (\ref{diagramma3}).
\end{itemize}
\end{enumerate}
\end{proof}

\section{Proofs from Section~\ref{sec:sat}}

\vskip 2mm\noindent
\textbf{Lemma~\ref{lem:sat1}}\emph{
 Let $\phi$ be a quantifier-free formula; then it is possible to compute 
 in linear time
 a finite 
 separation pair $\Phi=(\Phi_1, \Phi_2)$
 such that $\phi$ is \AXDTI-satisfiable iff so is  $ \Phi$.
}
\vskip 1mm
\begin{proof} We first flatten all atoms from $\phi$ by repeatedly abstracting out subterms 
 (to abstract out a subterm $t$, we introduce a fresh variable $x$ and update $\phi$ to
 $x=t\wedge  \phi(x/t)$); then we remove all atoms of the kind $a=b$ occurring in $\phi$ by replacing them by the equivalent formula~\eqref{eq:eleq}, namely
 $$
  \diff(a,b)=0 \wedge rd(a,0)=rd(b,0)~.
 $$
 Then we abstract out all terms of the kind $wr(b,i,e), \diff(a,b)$ and $\len{a}$, so that $\phi$ has now the form
 $\Phi_1\wedge \Phi_2$, where $\Phi_2$ does not contain $wr, \diff, \len{-}$-symbols and $\Phi_1$ is a
 conjunction of atoms of the form $a=wr(b,i,e), i=\diff(a,b), j=\len{a}$.
 Finally, we add to $\Phi_1$ the missing atoms of the kind $\len{a}=i$  required by Definition~\ref{def:separated}.\footnote{
 The transformation of Lemma~\ref{lem:sat1} does not introduce in $\Phi_1$ any formula of the kind $\diff_n(a,b)= k_n$ (for $n>1$).
 These \formulae\ will however be introduced by the Step 1 of the interpolation algorithm of Section~\ref{sec:algo}.
 }
 \end{proof}

\vskip 2mm\noindent
\textbf{Lemma~\ref{lem:sat2}} \emph{
 The following conditions are equivalent for a finite 0-instantiated separation pair $\Phi=(\Phi_1, \Phi_2)$:  
 \begin{compactenum}
 \item[{\rm (i)}] $\Phi$ is \AXDTI-satisfiable;
 \item[{\rm (ii)}] $\bigwedge \Phi_2$ is $T_I\cup \EUF$-satisfiable.
 \end{compactenum}
}
\vskip 1mm
\begin{proof}
$(i)\Rightarrow(ii)$ is clear.

\noindent 
To prove 
$(ii)\Rightarrow(i)$, let $\cA$
be the model witnessing the satisfiability of
 $\bigwedge\Phi_2$ in $T_I\cup \mathcal{EUF}$
 and let $\cI$ be the set of all index variables occurring in $\Phi_1\cup \Phi_2$. According to the definition of a separated pair, for every array variable $a$ occurring in 
 $\Phi_1\cup \Phi_2$ there is an index variable
 $l_a$ such that:
$$
\cA \models f_a(i)\neq \bot \leftrightarrow 0\le i \le l_a.
$$
for all $i\in \cI$ (here $f_a$ is the unary function symbol replacing $a$ in $T_I\cup \mathcal{EUF}$). 

The \emph{standardization} $\cA'$ of $\cA$ is the $T_I\cup \mathcal{EUF}$-model obtained from $\cA$ by 
modifying the values $a(k)$ (for all array variables $a$ occurring in $\Phi_2$ and for all indexes $k\in \INDEX^\cA$ different from the elements assigned in $\cA$ to the variables  in $\cI$) in such a way that we have 
$$
   \cA'\models f_a(k)=\bot \leftrightarrow (l_a<k \lor k<0),
$$
$$
  \cA'\models  f_a(k)=el \leftrightarrow 0\le k \le l_a~.
$$
The standardization $\cA'$ of $\cA$ is still a model of $\bigwedge\Phi_2$.
However, in $\cA'$ now also the formulae (\ref{lunghequivt3}), (\ref{diffiterate}) and (\ref{wrt3}) hold: this is because the $\cI$-instantiations of the universal index quantifiers occurring in such formulae
were taken care in $\cA$ and their truth value is not modified passing to $\cA'$, whereas the construction of $\cA'$ takes care of the instantiations outside $\cI$.

\noindent 
Let us now define an \AXDTI-model $\cM$ satisfying $\Phi$. We first build a structure $\cN$ where $\diff$ may not be totally defined. 
We let $\INDEX^{\cN}=\INDEX^{\cA'}$ and $\ELEM^{\cN}=\ELEM^{\cA'}$; 
we take as  
 $\ARRAY^\cN$ 
 the set of all positive-support functions from $\INDEX^{\cN}$ into 
 $\ELEM^{\cN}$: this includes all functions of the form $f_a$.
 In addition, if $\bigwedge_{n=1}^l\diff_n(a_1, a_2)=k_n\in \Phi_1$, then the related iterated maxdiff's are defined in $\cN$ and we have $\cN\models \bigwedge_{n=1}^l\diff_n(a_1, a_2)=k_n$
 by the above construction. 
 Thus $\Phi$ holds in $\cN$ and
 in order to obtain our final $\cM$ we only need to apply 
 Theorem~\ref{thm:extension}.
\end{proof}

\section{Proofs from Section~\ref{sec:algo}}

\vskip 2mm\noindent
\textbf{Theorem~\ref{thm:al}}\emph{
 The above rule-based algorithm computes a quantifier-free interpolant for every 
 \AXDTI-mutually unsatifiable pair $A^0, B^0$ of quantifier-free formulae.
}
\vskip 1 mm
\begin{proof}
We only need to prove that Step 3 really applies. 

Suppose not;
let $A=(A_1, A_2)$ and 
$B=(B_1, B_2)$ be the separated pairs obtained  
after  applications of Steps 1 and 2. If Step 3 does not apply, then $A_2\wedge B_2$ is $T_I\cup \EUF$-consistent.
  We claim that  $(A, B)$ is \AXDTI-consistent (contradicting 
  that $(A^0, B^0)\subseteq (A, B)$ was  
\AXDTI-inconsistent).

Let $\cM$ be a $T_I\cup \EUF$-model of $A_2\wedge B_2$.
$\cM$ is a two-sorted structure (the sorts are \INDEX and \ELEM) endowed for every array  constant $d$ occurring in $A\cup B$ of a function $d^\cM:\INDEX^\cM \longrightarrow \ELEM^\cM$. In addition, $\INDEX^\cM$ is a model of $T_I$. 
 We list the properties of $\cM$ that comes from the fact that our Steps 1-2 have been  applied
 (below, we denote by $k^\cM$ the element of $\INDEX^\cM$ assigned to an index constant $k$):\footnote{
 Thus if e.g. $k,l$ are index constants, $\cM\models k=l$ is the same as $k^\cM=l^\cM$.
 }
\begin{description}
 \item[{\rm (a)}] we have that $\cM\models \bigwedge A_1(\cI_A)$ (where $\cI_A$ is the set of \alocal constants) and $\cM\models \bigwedge B_1(\cI_B)$ (where $\cI_A$ is the set of \blocal constants): this is because $(A_1, A_2)$ and $(B_1, B_2)$ are 0-instantiated by Step 2;\footnote{ The sets $A_1(\cI_A),  B_1(\cI_B)$
 are introduced in Definition~\ref{def:instsep}.
 }
\item[{\rm (b)}] 
for \abcommon array variables $c_1, c_2$, we have that $A_1\cap B_1$ contains a literal of the kind $\diff_n(c_1, c_2)= k_n$ for $n\leq N$; suppose that $\cM \models l_{c_1}= l_{c_2}$\footnote{ Recall that  $l_{c_1}, l_{c_2}$ are the \abcommon constants such that the literals 
 $\len{c_1}=l_{c_1}, \len{c_2}=l_{c_2}$ belongs to $A_1\cap B_1$.} and that $k$ is an index constant such that 
 $\cM\models k\neq l$ for all \abcommon index constant $l$; then, we can have $\cM \models c_1(k)\neq c_2(k)$   only when $\cM \models k< k_N$: this is because Step 1 has been  applied and because of (a).
\end{description}

We expand $\cM$ to an \AXEXTTI-structure $\cN$ and endow it with an assignment to our \alocal and \blocal variables, in such a way that all \diff\ operators mentioned in $A_1, B_1$ are defined and all formulae in $A, B$ are true. In view of Theorem~\ref{thm:extension}, this structure can be expanded to the desired full model of \AXDTI.
We take $\INDEX^\cN$ and $\ELEM^\cN$ to be equal to $\INDEX^\cM$ and $\ELEM^\cM$; the $T_I$-reduct of $\cN$ will be equal to the $T_I$-reduct of $\cM$ and we let $x^\cN=x^\cM$ for all index and element constants occurring in $A\cup B$. $\ARRAY^\cN$ is the set of all positive support functions from $\INDEX^\cN$ into $\ELEM^\cN$. 
The interpretation of \alocal and \blocal constants of sort \ARRAY is more subtle. We need a d\'etour to explain it.

Let $k$ be an  index constant such that 
 $\cM\models k\neq l$ for all \alocal index constant $l$;
we introduce an equivalence relation  
$\equiv_k$ on the set of \alocal array variables as follows: $\equiv_k$ is the smallest equivalence relation that contains all pairs $(a_1,a_2)$ such that $\cM\models l_{a_1}=l_{a_2}$ and moreover an atom of one of the following two 
kinds belongs to $A^0_1$: (I) $a_1= wr(a_2, i, e)$; (II) $\diff(a_1,a_2)= l$, for an $l$ such that $\cM\models l<k$.

\vskip 2mm\noindent
\textbf{Claim}: \emph{if $c_1, c_2$ are \abcommon and $c_1\equiv_k c_2$, then $c_1^\cM(k^\cM)=c_2^\cM(k^\cM)$.}

\vskip 1mm\noindent
\emph{Proof of the Claim.} 
The claim is proved 
by preventively showing, for every \alocal constants $a_1,a_2$ such that $a_1\equiv_k a_2$,
 that the number of the \alocal constants $j$
such that $\cM\models k<j$ and $a_1^\cM(j^\cM)\neq a_2^\cM(j^\cM)$ is less or equal to $N_A<N$.
This is easily shown 
by induction on the length of the finite sequence witnessing 
$a_1\equiv_k a_2$.\footnote{ According to the definition of reflexive-symmetric-transitive closure, if 
$a_1\equiv_k a_2$ holds then there are $d_0, \dots, d_n$ such that $d_0=a_1, d_n=a_2$ and for each $j<n$,
we have that either $(d_j, d_{j+1})$ or $(d_{j+1}, d_{j})$ satifies the above requirements: the induction is on such $n$.
Notice that the statement is not entirely obvious because the number of the \alocal index constants is much bigger than $N_A$ (for instance, it includes the \abcommon constants introduced in Step 1). However 
 induction is easy: it goes through the atoms occurring in the input set $A^0_1$ and uses (a).
 The required observations  are the following: if $d_j=wr(d_{j+1},i,e)\in A^0_1$, then  the only \alocal constant where $d_j^\cM$ and $d_{j+1}^\cM$ can differ is $i$; if $\diff(d_j,d_{j+1})=l\in A^0_1$ and $\cM \models l< k$, then  $d_j^\cM$ and $d_{j+1}^\cM$ cannot differ on any \alocal constant above $k^\cM$. Iterating these observations during induction, the claim is clear: we can collect at most the set of the \alocal constants occurring in $A^0_1$ within a $wr$ symbol.
 }
We apply this observation to the \abcommon array variables $c_1, c_2$ and let us consider the atoms 
 $\diff_1(c_1, c_2)=k_1, \dots, \diff_n(c_1, c_2)=k_N\in A_1$. 
 Now $k_1, \dots, k_N$ are all \abcommon (hence also \alocal) constants, moreover  $N> N_A$ and the number of the \alocal constants 
 above $k^\cM$ where $c_1, c_2$ differ is at most $N_A$. According to (b) above, if for absurdity  $c_1^\cM(k^\cM)= c_2^\cM(k^\cM)$ does not hold, then 
 we have $\cM \models k < k_N$. Since $\cM \models k\geq 0$
 (otherwise $c_1^\cM(k^\cM)= c_2^\cM(k^\cM)$ follows),
 this means by Lemma~\ref{lem:elim} that we have $\cM\models k_1>\cdots  > k_N>0$.  
 However $k_1, \dots, k_N$ are all \alocal constants above $k$, their number is bigger than $N_A$, 
 hence we must have 
 $\cM \models c_1^\cM(k_i^\cM)=c_2^\cM(k_i^\cM)$ for some $i=1,\dots,N$; the latter  implies $\cM \models k_i=\dots=k_N=0$ by Lemma~\ref{lem:elim},
 absurd.
\eop
\vskip 2mm

In order to interpret \alocal constants of sort \ARRAY, we assign to an \alocal constant $a$ of sort \ARRAY the function $a^\cN$ defined as follows for every $i\in \INDEX^\cN$:
\begin{compactenum}[(\dag-i)]
\item if $i$ is equal to $k^\cM$, where $k$ is an \alocal constant, then $a^\cN(i):=a^\cM(k^\cM)$;
\item if $i$ is different from $k^\cM$ for every \alocal constant $k$, but nevertherless $i$ is equal to $k^\cM$
for some (necessarily \bstrict) index constant $k$  and 
 there is an \abcommon array variable $c$ such that $c\equiv_k a$,
then $a^\cN(i)$ is equal to $c^\cM(k^\cM)$;
\item
in the remaining cases,
 $a^\cN(i)$ is equal to $el^\cM$ or $\bot^\cM$ depending whether 
$\cM\models 0\leq i \wedge i\leq l_a$ holds or not.
\end{compactenum}

Notice that 
$a^\cN(i)$ is univocally specified in case ($\dag$-ii) because of the above claim.
We now show that 
\begin{itemize}
 \item[(*)] \emph{all $a^\cN$ are positive-support functions and  all \formulae from $A_1\cup A_2$ are true in $\cN$}.
\end{itemize}
Recall in fact that \formulae in $A_2$ are Boolean combinations of \alocal atoms of the kind~\eqref{phi2}: these are $T_I\cup \EUF$-atoms and, due to their shape, each of them is true in $\cM$ iff it is true in $\cA$ (this is because $rd$ functions are applied only to \alocal index constants, so that the modifications we introduced for passing from $a^\cM$ to $a^\cN$ does not affect truth of these atoms).    Concerning \formulae in $A_1$, these are 
 all $wr, \diff$ and $ \len{-}$-atoms.\footnote{ In addition, we have $\diff_n$-atoms for $n>1$, but these are all \abcommon atoms that were not part of the initial pair $A^0, B^0$. In fact they are also true in $\cN$, but strictly speaking we do not need to check this fact to get the absurdity that $A^0\wedge B^0$ is \AXDTI-consistent.}
 The reason why they are true in $\cN$ is the 0-instantiation performed by Step 2 (see (a) above).
 For example, consider an atom of the kind $\len{a}=l_a$ appearing in $A_1$: since formulae~\eqref{lunghequivt3}  
 have been instantiated with all the \alocal index constants  via Step 2,  for all \alocal index constant $h$, we have 
$\cM\models rd(a,h)\neq \bot$ iff $\cM\models 0\le h \le i$. 
Now, 
 thanks to definition ($\dag$),  for all the  elements $h\in \INDEX^\cN$ we have an analogous result: in case $h$ is equal to $k^\cM$ for some \alocal constant $k$, we employ definition ($\dag$-i), otherwise we use ($\dag$-ii) or ($\dag$-iii) (for 
 ($\dag$-ii), notice that $a\equiv_k c$ implies that $\cM\models l_a=l_c$ and 0-instantiation
 guarantees that $c^\cM(k^\cM)$ is equal to $\bot^\cM$ or to $\el^\cM$ depending whether 
 $\cM\models 0\leq k \wedge k\leq  l_c$ holds or not).
 Hence we get that $\cA$ satisfies formula~\eqref{lunghequivt3}, 
 since the universal quantifier has been instantiated in all possible ways. Thus $\cA\models \len{a} = l_a$, by Lemma~\ref{lem:univinst}. The other cases are similar: notice in particular that if $a=wr(a',i,e)\in A_1$ 
 then $a\equiv_k a'$. If $\diff(a, a')=i\in A_1$, the relevant case is when $\cM\models i<k$ and $\cM\models l_{a_1}=l_{a_2}$, but in this case we have $a\equiv_k a'$ too.

The  assignments to the \blocal array variables $b$ are  defined analogously, so that 
\begin{itemize}
 \item[(*)] \emph{all $b^\cN$ are positive-support functions and  all \formulae from $B_1\cup B_2$ are true in $\cN$}.
\end{itemize}

There is however one important point to notice:
 for all \abcommon constants $c$ of sort \ARRAY,
 our specification of  $c^\cM$ \emph{does not depend} on the fact that we use the above definition for \alocal or for \blocal array variables: to see this,
 we only have to notice that $c\equiv_k c$ holds in case ($\dag$-ii) is applied. This remark concludes our proof.
\end{proof}



%
%
%
\end{document}